%% file: attackSynthesis.tex
\documentclass[runningheads]{llncs}

\usepackage[normalem]{ulem}
\newcommand{\add}[1]{#1}
\newcommand{\rem}[1]{}
\newcommand{\noted}[1]{}

\title{Automated Attacker Synthesis for Distributed Protocols}
\subtitle{\add{Expanded with Erratum}}
\author{
        Max von Hippel \Letter%\inst{1}
        \and
        Cole Vick%\inst{1}
        \and
        Stavros Tripakis \Letter%\inst{1}
        \and
        Cristina Nita-Rotaru \Letter%\inst{1}
}
\date{\today}
\authorrunning{von Hippel et al.}
\institute{Northeastern University, Boston MA 02118, USA\\
\email{\{vonhippel.m,stavros,c.nitarotaru\}@northeastern.edu}\\
\email{vick.c@husky.neu.edu}}

\let\proof\relax
\let\endproof\relax
\usepackage{
    graphicx,
    tikz,
    centernot,
    mathtools,
    adjustbox,
    fancyvrb,
    ifthen,
    times,
    float,
    multirow,
    tabularx,
    mdframed}
\newcolumntype{L}{>{\raggedright\arraybackslash}X}
\usepackage{url}
\usetikzlibrary{
    arrows,shapes,positioning,decorations.markings,
    decorations.pathreplacing,hobby}

\newenvironment{sketch}{%
  \proof}{\endproof}

\usepackage[font=footnotesize]{caption}

\DeclarePairedDelimiter{\abs}{\lvert}{\rvert}

\tikzset{
    looped/.style={
        decoration={markings,mark=at position 0.999 with {\arrow[scale=2]{>}}},
        postaction={decorate},
        >=stealth
    },
    straight/.style={
        decoration={markings,mark=at position 1 with {\arrow[scale=2]{>}}},
        postaction={decorate},
        >=stealth
    },
    loopedSF/.style={
        decoration={
            markings,
            mark=at position 0.999 with {\arrow[scale=2]{>}},
            mark=at position 0.5 with {\arrow[scale=2]{>}}},
        postaction={decorate},
        >=stealth
    },
    straightSF/.style={
        decoration={
            markings,
            mark=at position 0.999 with {\arrow[scale=2]{>}},
            mark=at position 0.5 with {\arrow[scale=2]{>}}},
        postaction={decorate},
        >=stealth
    },
    triangle/.style = {fill=white, draw=black, regular polygon, regular polygon sides=3 },
    node rotated/.style = {rotate=180},
    border rotated/.style = {shape border rotate=180}
}

\newcommand{\F}[0]{\mathbf{F}}
\newcommand{\U}[0]{\mathbf{U}}
\newcommand{\G}[0]{\mathbf{G}}
\newcommand{\X}[0]{\mathbf{X}}

\newcommand{\showtechreport}{1} % set to 0 to switch back from techreport to short version
\newcommand{\showfuturework}{0} % set to 0 to hide
\newcommand{\techreport}[2]{\ifthenelse{\equal{\showtechreport}{1}}{#1}{#2}}
\newcommand{\futurework}[1]{\ifthenelse{\equal{\showfuturework}{1}}{{\color{red}{#1}}}{}}

\techreport{
    \usepackage{fullpage,pifont,diagbox,xcolor,enumerate,amsthm,amsfonts,amssymb,amsmath}
    \newcommand{\xmark}{\ding{55}}
    \usepackage[most]{tcolorbox}
    \definecolor{lightgreen}{rgb}{0.56, 0.93, 0.56}
    \definecolor{moonstoneblue}{rgb}{0.45, 0.66, 0.76}
    \usepackage[ruled,vlined,linesnumbered]{algorithm2e}
    
    \SetCommentSty{mycommfont}
    \raggedbottom

}{}

\usepackage[misc]{ifsym}

\begin{document}
\maketitle

\begin{abstract}
Distributed protocols
should be robust to both benign malfunction (e.g. packet loss or delay) and attacks (e.g. message replay).
In this paper we take a formal approach to the automated synthesis of attackers, i.e. adversarial processes that can cause the protocol to malfunction.
Specifically, given a formal \textit{threat model} capturing the distributed protocol model and network topology,
as well as the placement, goals, and interface \techreport{(inputs and outputs)}{} of potential attackers, 
we automatically synthesize an attacker.
  We formalize four attacker synthesis problems - across attackers that always succeed versus those that sometimes fail, and attackers that may attack forever versus those that may not - and we propose algorithmic solutions to two of them.  
 We report on a prototype implementation called \textsc{Korg} and its application to TCP as a case-study.
  Our experiments show that \textsc{Korg} can automatically generate well-known attacks for TCP within seconds or minutes.
\end{abstract}

\techreport{}{\vspace{-1em}}
\section{Introduction}
\techreport{}{\vspace{-1em}}
\input{sections/Section1Introduction.tex}

\section{Formal Model Preliminaries}\label{sec:bk}
\techreport{}{\vspace{-1em}}
\input{sections/Section2Preliminaries.tex}

\techreport{}{\vspace{-1em}}

\section{Attacker Synthesis Problems}\label{sec:pb}
\techreport{}{\vspace{-1em}}
\input{sections/Section3AttackerSynthesisProblems.tex}

\techreport{}{\vspace{-2em}}

\section{Solutions}\label{sec:sol}
\techreport{}{\vspace{-1em}}
\input{sections/Section4Solutions.tex}

\section{Case Study: TCP}\label{sec:tcp}
\techreport{}{\vspace{-1em}}
\input{sections/Section5CaseStudyTCP.tex}

\techreport{}{\vspace{-0.5em}}

\section{Related Work}\label{sec:relwork}
\techreport{}{\vspace{-1em}}
\input{sections/Section6RelatedWork.tex}

\techreport{}{\vspace{-1.5em}}

\techreport{\section{Conclusions and Future Work}}{\section{Conclusion}}\label{sec:concl}
\techreport{}{\vspace{-1em}}
\input{sections/Section7Conclusion.tex}

\section{Erratum}\label{sec:err}
\techreport{}{\vspace{-1em}}
\input{sections/Section8Erratum.tex}

\techreport{\medskip}{}

\techreport{\noindent\rule{\textwidth}{1pt}}{}

{\footnotesize 
\noindent\textbf{Acknowledgments} This material is based upon work supported by the National Science Foundation under NSF SaTC award CNS-1801546. Any opinions, findings, and conclusions or recommendations expressed in this material are those of the author(s) and do not necessarily reflect the views of the National Science Foundation.  The authors thank four anonymous reviewers.
Additionally, the first author thanks Benjamin Quiring, Dr. Ming Li, and Dr. Frank von Hippel.

\techreport{\noindent\rule{\textwidth}{1pt}}{\vspace{-1em}}

\bibliographystyle{splncs04}
% \bibliographystyle{plain}
% \begingroup
% \let\clearpage\relax
\bibliography{attackSynthesis}
% \endgroup

\end{document}

%% file: sections/Section1Introduction.tex
Distributed protocols represent the fundamental communication backbone for all services
over the Internet.  Ensuring  the correctness and security of these protocols is critical
for the services built on top of them \cite{chong2016report}. 
Prior literature proposed different approaches to correctness assurance, e.g. testing \cite{Myers79,duran1981report}, or structural reasoning \cite{dijkstra1970notes}.
Many such approaches rely on manual analysis or are ad-hoc in nature.

In this paper, we take a systematic approach to the problem of security of distributed protocols, by using formal methods and synthesis \cite{alonzo57}.
Our focus is the automated generation of {\em attacks}.
But what exactly is an attack? The notion of an attack is often implicit in the
formal verification of security properties: it is a counterexample
violating some security specification. 
We build on this idea.
We provide a formal definition of {\em threat models} capturing the distributed protocol model and network topology, as well as the placement, goals, and capabilities of potential {\em attackers}.
%(e.g. Fig.~\ref{exampleThreatModels}).
\techreport{Note that an \emph{attacker goal} is simply the negation of a \emph{protocol property}, in the sense that the goal of an attacker is to violate desirable properties that the protocol must preserve.}{}
Intuitively, an attacker is a process that, when composed with the system, results a protocol property violation.

By formally defining attackers as processes, our approach has several benefits:
first, we can ensure that these processes are {\em executable}, meaning attackers are programs that reproduce attacks.
This is in contrast to other approaches that generate a {\em trace} exemplifying an attack, but not a program producing the attack, e.g. \cite{BlanchetCSFW01,NetSMC}.
Second, an explicit formal attacker definition allows us to distinguish different types of attackers, 
depending on: what exactly does it mean to violate a property (in some cases? in all cases?);
how the attacker can behave, etc.
We distinguish between
 $\exists$-attackers (that sometimes succeed in violating the security property) 
 and $\forall$-attackers (that always succeed); 
 and between attackers \emph{with recovery} (that eventually revert to normal system behavior) 
and attackers without (that may attack forever). We make four primary contributions.  
\begin{itemize}
\item We propose a novel formalization of threat models and attackers, where the threat models algebraically capture not only the attackers but also the attacker goals, the environmental and victim processes, and the network topology.  

\item We formalize four attacker synthesis problems -- $\exists$ASP, R-$\exists$ASP,  $\forall$ASP, R-$\forall$ASP --  one for each of the four combinations
of types of attackers.

\item We propose solutions for $\exists$ASP and R-$\exists$ASP via reduction to model-checking. 
The key idea of our approach is to replace the vulnerable processes - the victim(s) - by appropriate ``gadgets'', then ask a model-checker whether the resulting system violates a certain property.

\item We implement our solutions in a prototype open-source tool called \textsc{Korg}, 
and apply \textsc{Korg} to TCP connection establishment and tear-down routines.
Our experiments show \textsc{Korg} is able to automatically synthesize realistic, 
well-known attacks against TCP within seconds or minutes.
\end{itemize}

The rest of the paper is organized as follows. We present background material in 
\techreport{Section~}{$\mathsection$}\ref{sec:bk}. We define attacker synthesis problems in \techreport{Section~}{$\mathsection$}\ref{sec:pb} and
 present solutions in \techreport{Section~}{$\mathsection$}\ref{sec:sol}.
We describe the TCP case study in \techreport{Section~}{$\mathsection$}\ref{sec:tcp},
present related work in \techreport{Section~}{$\mathsection$}\ref{sec:relwork}, and conclude
in~\techreport{Section~}{$\mathsection$}\ref{sec:concl}.
\add{The solutions in \techreport{Section~}{$\mathsection$}\ref{sec:sol} are updated to fix a mathematical error in the previously published version of this document.  We explain that error in \techreport{Section~}{$\mathsection$}\ref{sec:err}.}

%% file: sections/Section2Preliminaries.tex
We model distributed protocols as interacting processes,  in the spirit of \cite{SIGACT17}.
We next define formally these processes and their composition.  \techreport{We also define formally the specification language that we use, namely LTL.}{}  
We use $2^X$ to denote the power-set of $X$, and
$\omega$-exponentiation to denote infinite repetition, e.g., $a^\omega=a a a \cdots$.

\subsection{Processes}

%\sloppy
\begin{sloppypar}
\begin{definition}[Process]
A \emph{process} is a tuple 
$P = \langle \text{AP}, I, O, S, s_0, T, L \rangle$ with set of \emph{atomic propositions} AP, 
set of \emph{inputs} $I$, set of \emph{outputs} $O$, set of \emph{states} $S$, 
\emph{initial state} $s_0 \in S$, \emph{transition relation} $T \subseteq S \times (I \cup O) \times S$, 
and (total) \emph{labeling function} $L : S \to 2^{\text{AP}}$, such that:
$\text{AP}, I, O$, and $S$ are finite; and $I \cap O = \emptyset$.
\end{definition}
\end{sloppypar}

\techreport{
In formal methods, {\em Kripke Structures} \cite{Kripke63} are commonly used to describe computer programs, because they are automata (and so well-suited to describing computer programs) and their states are labeled with atomic propositions (so they are well-suited to modal logic).  A {\em process} is just a Kripke Structure with {\em inputs} and {\em outputs}.  Using Kripke Structures allows us to leverage LTL for free, and separating messages into inputs and outputs allows us to describe network topologies entirely using just the interfaces of the interacting processes.  This idea is fundamental to our formalism of threat models in Section~\ref{sec:sol}.  We now explain the technical details of processes.}{}

Let $P = \langle \text{AP}, I, O, S, s_0, T, L \rangle$ be a process.  For each state~$s \in S$,~$L(s)$ is a subset of AP containing the atomic propositions that are true at state~$s$. 
Consider a transition~$(s,x,s')$ starting at state~$s$ and ending at state~$s'$ with label~$x$.
If the label $x$ is an input, then the transition is called an \emph{input transition} and denoted $s \xrightarrow[]{x?} s'$.  Otherwise,~$x$ is an output, and the transition is called an \emph{output transition} and denoted $s \xrightarrow[]{x!} s'$.  A transition $(s, x, s')$ is called \emph{outgoing} from state $s$ and \emph{incoming} to state~$s'$.

A state $s \in S$ is called a \emph{deadlock} iff it has no outgoing transitions.
\techreport{The state $s$ is called \emph{reachable} if either it is the initial state or there exists a sequence of transitions $\big( (s_i, x_i, s_{i+1}) \big)_{i = 0}^m \subseteq T$ starting at the initial state $s_0$ and ending at $s_{m + 1} = s$.  Otherwise, $s$ is called \emph{unreachable}.}{}
The state $s$ is called \emph{input-enabled} iff, for all inputs $x \in I$, there exists some state $s' \in S$ such that there exists a transition $(s, x, s') \in T$. 
We call $s$ an \emph{input} state (or \emph{output} sate) if all its outgoing transitions are input transitions (or output transitions, respectively).
\techreport{States with both outgoing input transitions and outgoing output transitions are neither input nor output states, while states with no outgoing transitions (i.e., deadlocks) are (vacuously) both input and output states.}{}

\techreport{Various definitions of process determinism exist; ours is a variation on that of \cite{SIGACT17}.}{}
A process $P$ is \emph{deterministic} iff all of the following hold: (i) its transition relation $T$ can be expressed as a \add{(possibly partial)} function $S \times (I \cup O) \to S$; (ii) every non-deadlock state in $S$ is either an input state or an output state, but not both; (iii) input states are input-enabled; and (iv) each output state has only one outgoing transition. 
%; and if $I$ is non-empty then there is always a reachable input state.
Determinism guarantees that: each state is a deadlock, an input state, or an output state; when a process outputs, its output is uniquely determined by its state; and when a process inputs, the input and state uniquely determine where the process transitions.  
\techreport{More intuitively, deterministic processes can be translated into concrete programs in languages like \textsc{C} or \textsc{Java}.  Determinism therefore helps us make our attackers realistic.}{}

A \emph{run} of a process $P$ is an infinite sequence $r = \big( (s_i, x_i, s_{i+1}) \big)_{i = 0}^{\infty} \subseteq T^{\omega}$ of consecutive transitions.  We use $\text{runs}(P)$ to denote all the runs of $P$.  A run over states $s_0,s_1,...$ induces a sequence of labels $L(s_0),L(s_1),...$ called a \textit{computation}.  
\techreport{Given a (zero-indexed) sequence $\nu$, we let $\nu[i]$ denote the $i^{\text{th}}$ element of $\nu$; $\nu[i:j]$, where $i \leq j$, denote the finite infix $(\nu[t])_{t = i}^{j}$; and $\nu[i:]$ denote the infinite postfix $(\nu[t])_{t = i}^{\infty}$; we will use this notation for runs and computations.

Given two processes \(P_i = \langle \text{AP}_i, I_i, O_i, S_i, s_0^i, T_i, L_i \rangle\) for $i = 1, 2$, we say that $P_1$ is a {\em subprocess} of $P_2$, denoted $P_1 \subseteq P_2$, if $\text{AP}_1 \subseteq \text{AP}_2, I_1 \subseteq I_2, O_1 \subseteq O_2, S_1 \subseteq S_2, T_1 \subseteq T_2,$ and, for all $s \in S_1$, $L_1(s) \subseteq L_2(s)$.
}{}

\subsection{Composition}
The composition of two processes $P_1$ and $P_2$ is another process denoted $P_1 \parallel P_2$, capturing both the individual behaviors of $P_1$ and $P_2$ as well as their interactions with one another\techreport{ (e.g. Fig.~\ref{exampleComposition}).}{.}  We define the asynchronous parallel composition operator $\parallel$ with rendezvous communication as in \cite{SIGACT17}.
\begin{definition}[Process Composition]
Let \(P_i = \langle \text{AP}_i, I_i, O_i, S_i, s_0^i, T_i, L_i \rangle\) be processes,
for $i = 1, 2$.  For the composition of $P_1$ and $P_2$ 
(denoted $P_1 \parallel P_2$) to be well-defined,
\techreport{
	the processes must have no common outputs,
	i.e., $O_1 \cap O_2 = \emptyset$, 
	and no common atomic propositions,
	i.e., $\text{AP}_1 \cap \text{AP}_2 = \emptyset$.

}{
	the processes must have no common outputs, 
	and no common atomic propositions.
}
Then $P_1 \parallel P_2$ is defined below:
\begin{equation}
P_1 \parallel P_2 
= \langle 
\text{AP}_1 \cup \text{AP}_2, 
(I_1 \cup I_2) \setminus (O_1 \cup O_2), 
O_1 \cup O_2, 
S_1 \times S_2, 
(s_0^1, s_0^2), 
T, 
L \rangle
\end{equation}
... where the transition relation $T$ is precisely the set of transitions $(s_1, s_2) \xrightarrow[]{x} (s_1', s_2')$ such that, for $i = 1, 2$, if the label $x \in I_i \cup O_i$ is a label of $P_i$, then $s_i \xrightarrow[]{x} s_i' \in T_i$, else $s_i = s_i'$.  $L : S_1 \times S_2 \to 2^{\text{AP}_1 \cup \text{AP}_2}$ is the function defined as $L(s_1, s_2) = L_1(s_1) \cup L_2(s_2)$.
\end{definition}

\techreport{Intuitively, we define process composition to capture two primary ideas: (1) \emph{rendezvous communication}, meaning that a message is sent at the same time that it is received, and (2) \emph{multi-casting}, meaning that a single message could be sent to multiple parties at once. We can use so-called \emph{channel} processes to build asynchronous communication out of rendezvous communication (as we do in Section~\ref{sec:tcp}), and we can easily preclude multi-casting by manipulating process interfaces.  Our definition therefore allows for a variety of communication models, making it flexible for diverse research problems.  We next explain and illustrate the technical details.}{}

\techreport{A state of the composite process $P_1 \parallel P_2$ is a pair $(s_1, s_2)$ consisting of a state $s_1 \in S_1$ of $P_1$ and a state $s_2 \in S_2$ of $P_2$.  The initial state of $P_1 \parallel P_2$ is a pair $(s_0^1, s_0^2)$ consisting of the initial state $s_0^1$ of $P_1$ and the initial state $s_0^2$ of $P_2$.  The inputs of the composite process are all the inputs of $P_1$ that are not outputs of $P_2$, and all the inputs of $P_2$ that are not outputs of $P_1$.  The outputs of the composite process are the outputs of the individual processes.
$P_1 \parallel P_2$ has three kinds of transitions $(s_1, s_2) \xrightarrow[]{z} (s_1', s_2')$.  In the first case, $P_1$ may issue an output $z$.  If this output $z$ is an input of $P_2$, then $P_1$ and $P_2$ move simultaneously and $P_1 \parallel P_2$ outputs $z$.  Otherwise, $P_1$ moves, outputting $z$, but $P_2$ stays still (so $s_2 = s_2'$).  The second case is symmetric to the first, except that $P_2$ issues the output.  In the third case, $z$ is neither an output for $P_1$ nor for $P_2$.  If $z$ is an input for both, then they synchronize.  Otherwise, whichever process has $z$ as an input moves, while the other stays still.}{}

\techreport{
\begin{figure}[H]
\centering
\begin{adjustbox}{max totalsize={1.0\textwidth}{.5\textheight}}
\begin{tikzpicture}
\node[] (empty) {};
\node[draw,circle] (s0) [above=of empty] {\large $s_0 : \emptyset$};
\node[draw,circle] (s1) [right=of s0] {\large $\begin{aligned}s_1:\\\{ p, q \}\end{aligned}$};
\node[draw,circle] (s2) [right=of s1] {\large $\begin{aligned}s_2:\\\{ q \}\end{aligned}$};
\draw[straight] (empty) to (s0);
\draw[straight] (s2) to[above] node {\large $x?$} (s1);
\draw[looped] (s0) to[above,out=north east,in=north west,looseness=1] node {\large $x?$} (s1);
\draw[looped] (s0) to[below,out=south east,in=south west,looseness=1] node {\large $z?$} (s1);
\draw[looped] (s2) to[above,out=north west,in=north east,looseness=4] node {\large $v!$} (s2);
\draw[looped] (s0) to[above,out=north east,in=north west,looseness=4] node {\large $w!$} (s0);

\node[] (comp) [right=of s2] {\Huge $\parallel$};

\node[draw,circle] (q0) [right=of comp] {\large $q_0 : \{ r \}$};
\node[] (empty2) [below=of q0] {};
\node[draw,circle] (q1) [right=of q0] {\large $q_1 : \emptyset$};
\draw[straight] (empty2) to (q0);
\draw[looped] (q0) to[out=north east,in=north west,above] node {\large $x!$} (q1);
\draw[looped] (q1) to[out=south west,in=south east,below] node {\large $w?$} (q0);
\draw[looped] (q1) to[out=north east,in=south east,right,looseness=4] node {\large $m!$} (q1);

\node[] (equals) [right=of q1] {\Huge \, $=$};

\node[] (emptyFinal) [above=of equals] {};

\node[draw,circle] (s0q0) [right=of emptyFinal] 
	{\large $\begin{aligned}(s_0,q_0):\\\{ r \}\end{aligned}$};
\node[draw,circle] (s1q0) [right=of s0q0]
	{\large $\begin{aligned}(s_1,q_0):\\\{ p, q, r \}\end{aligned}$};
\node[draw,circle] (s2q0) [right=of s1q0]
	{\large $\begin{aligned}(s_2,q_0):\\\{ q, r \}\end{aligned}$};
\node[draw,circle] (s0q1) [below=of s0q0]
	{\large $\begin{aligned}(s_0,q_1):\\\emptyset\end{aligned}$};
\node[draw,circle] (s1q1) [right=of s0q1]
	{\large $\begin{aligned}(s_1,q_1):\\\{ p, q \}\end{aligned}$};
\node[draw,circle] (s2q1) [right=of s1q1]
	{\large $\begin{aligned}(s_2,q_1):\\\{ q \}\end{aligned}$};

\draw[straight] (emptyFinal) to (s0q0);

\draw[straight] (s0q0) to[above] node {\large $z?$} (s1q0);
\draw[straight] (s0q0) to[above] node {\large $\, \, x!$} (s1q1);
\draw[looped] (s2q0) to[out=north east,in=south east,looseness=4,right] node {\large $v!$} (s2q0);
\draw[straight] (s0q1) to[left] node {\large $w!$} (s0q0);
\draw[looped] (s2q1) to[out=north,in=east,looseness=4,right] node {\large $v!$} (s2q1);
\draw[straight] (s2q0) to[above] node {\large $x! \, \,$} (s1q1);
\draw[straight] (s0q1) to[above] node {\large $z?$} (s1q1);

\draw[looped] (s0q1) to[out=south east,in=south west,below,looseness=4] node {\large $m!$} (s0q1);
\draw[looped] (s1q1) to[out=south east,in=south west,below,looseness=4] node {\large $m!$} (s1q1);
\draw[looped] (s2q1) to[out=south east,in=south west,below,looseness=4] node {\large $m!$} (s2q1);
\end{tikzpicture}
\end{adjustbox}
\caption{Left is a process $P$ with atomic propositions $\text{AP} = \{ p, q \}$, inputs $I = \{ x, z \},$ outputs $O = \{ v, w \},$ states $S = \{ s_0, s_1, s_2 \},$ transition relation $T = \{ (s_0, w, s_0), (s_0, x, s_1), (s_0, z, s_1), (s_2, x, s_1), (s_2, v, s_2) \},$ and labeling function $L$ where $L(s_0) = \emptyset$, $L(s_1) = \{ p, q \},$ and $L(s_2) = \{ q \}$.  
Center is a process $Q = \langle \{ r \}, \{ w \}, \{ x, m \}, \{ q_0, q_1 \}, q_0, \{ (q_0, x, q_1), (q_1, m, q_1), (q_1, w, q_0) \}, L_Q \rangle$ where $L_Q(q_0) = \{ r \}$ and $L_Q(q_1) = \emptyset$.  Processes $P$ and $Q$ have neither common atomic propositions ($\{ p, q \} \cap \{ r \} = \emptyset$), nor common outputs ($\{ w, v \} \cap \{ x, m \} = \emptyset$), so the composition $P \parallel Q$ is well-defined.  Right is the process $P \parallel Q$.  Although $P \parallel Q$ is rather complicated, its only reachable states are $(s_0, q_0), (s_1, q_0),$ and $(s_1, q_1)$, and its only run is $r = \big( (s_0, q_0), x, (s_1, q_1) \big), \big( (s_1, q_1), m, (s_1, q_1) \big)^{\omega}$.  Non-obviously, the only computation of $P \parallel Q$ is $\sigma = \{ r \}, \{ p, q \}^{\omega}$.}
\label{exampleComposition}
\end{figure}
}{}

\techreport{Note that sometimes rendezvous composition is defined to match $s_1 \xrightarrow[]{z?} s_1'$ with $s_2 \xrightarrow[]{z!} s_2'$ to form a {\em silent} transition $(s_1, s_2) \xrightarrow[]{} (s_1', s_2')$, but with our definition the output is preserved, so the composite transition would be $(s_1, s_2) \xrightarrow[]{z!} (s_1', s_2')$.  This allows for \emph{multi-casting}, where an output event of one process can synchronize with multiple input events from multiple other processes.  It also means there are no silent transitions.}{}

The labeling function $L$ is total as $L_1$ and $L_2$ are total. Since we required the processes $P_1, P_2$ to have disjoint sets of atomic propositions, $L$ does not change the logic of the two processes under composition.
Additionally, $\parallel$ is commutative and associative \cite{SIGACT17}.

\subsection{LTL} 
\techreport{
	\emph{LTL} \cite{IEEEASFCS77} is a linear temporal logic for reasoning about computations.  In this work, we use LTL to formulate properties of processes.  The syntax of LTL is defined by the following grammar:
	\begin{equation}
	\phi ::= \underbrace{p \mid q \mid ...}_{\in \text{AP}} \mid \phi_1 \land \phi_2 \mid \neg \phi_1 \mid \X \phi_1 \mid \phi_1 \U \phi_2
	\end{equation}
	... where $p, q, ... \in \text{AP}$ can be any atomic propositions, and $\phi_1, \phi_2$ can be any LTL formulae.  Let $\sigma$ be a computation of a process $P$.  If an LTL formula $\phi$ is true about $\sigma$, we write $\sigma \models \phi$.  On the other hand, if $\neg ( \sigma \models \phi )$, then we write $\sigma \centernot{\models} \phi$.  The semantics of LTL with respect to $\sigma$ are as follows.
	\begin{equation}
	\begin{array}{lcl}
	\sigma \models p & \text{ iff } & p \in \sigma[0] \\
	\sigma \models \phi_1 \land \phi_2 & \text{ iff } & \sigma \models \phi_1 \text{ and } \sigma \models \phi_2 \\
	\sigma \models \neg \phi_1 & \text{ iff } & \sigma \centernot{\models} \phi_1 \\
	\sigma \models \X \phi_1 & \text{ iff } & \sigma[1:] \models \phi_1 \\
	\sigma \models \phi_1 \U \phi_2 & \text{ iff } & \big( \exists \, K \geq 0 \, : \, \sigma[K:] \models \phi_2 \text{, and }  \\
									&              & \, \, \, \forall \, 0 \leq j < K \, : \, \sigma[j:] \models \phi_1 \big) \\
	\end{array}
	\end{equation}

	Essentially, $p$ holds iff it holds at the first step of the computation; the conjunction of two formulae holds if both formulae hold; the negation of a formula holds if the formula does not hold; $\X \phi_1$ holds if $\phi_1$ holds in the next step of the computation; and $\phi_1 \U \phi_2$ holds if $\phi_2$ holds at some future step of the computation, and until then, $\phi_1$ holds.  Standard syntactic sugar include $\lor$, \textbf{true}, \textbf{false}, $\F$, $\G$, and $\to$.  For all LTL formulae $\phi_1, \phi_2$ and atomic propositions~$p \in \text{AP}$: $\phi_1 \lor \phi_2 \equiv \neg (\neg \phi_1 \land \neg \phi_2)$; $\textbf{true} \equiv p \lor \neg p$; $\textbf{false} \equiv \neg \textbf{true}$; $\F \phi_1 \equiv \textbf{true} \U \phi_1$; $\G \phi_1\equiv \neg \F \neg \phi_1$; and $\phi_1 \to \phi_2 \equiv (\neg \phi_1) \lor (\phi_1 \land \phi_2)$.  

	Example formulae include: 
	\begin{itemize}
		\item Lunch will be ready in a moment: $\X \texttt{lunch-ready}$.
		\item I always eventually sleep: $\G \F \texttt{sleep}$.
		\item I am hungry until I eat: $\texttt{hungry} \U \texttt{eat}$.
		\item $A$ and $B$ are never simultaneously in their \texttt{crit} states: $\G \neg (\texttt{crit}_A \land \texttt{crit}_B)$.
	\end{itemize}
}{
	\emph{LTL} \cite{IEEEASFCS77} is a linear temporal logic for reasoning about computations.  In this work, we use LTL to formulate properties of processes.  The syntax of LTL is defined by the following grammar: \(
	\phi ::= p \mid q \mid ... \mid \phi_1 \land \phi_2 \mid \neg \phi_1 \mid \X \phi_1 \mid \phi_1 \U \phi_2\), where the $p \mid q \mid ...$ are any atomic propositions $\in \text{AP}$, and $\phi_1, \phi_2$ can be any LTL formulae.

	Let $\sigma$ be a computation.  If $\sigma$ satisfies an LTL formula $\phi$ we write $\sigma \models \phi$.  If $\neg ( \sigma \models \phi )$, then we write $\sigma \centernot{\models} \phi$.  The satisfaction relation for LTL is formally defined as follows:
	$\sigma \models p$ if $p$ is true in $\sigma(0)$; $\sigma \models \X p$ if $p$ is true in $\sigma(1)$; $\sigma \models \F p$ if there exists some $K \geq 0$ such that $p$ is true in $\sigma(K)$; $\sigma \models \G p$ if for all $K \geq 0$, $p$ is true in $\sigma(K)$; $\sigma \models p \U q$ if there exists some $K \geq 0$ such that for all $k_1 < K \leq k_2$, $p$ is true in $\sigma(k_1)$ and $q$ is true in $\sigma(q_2)$; and $\sigma \models \phi_1 \land \phi_2$ if $\sigma \models \phi_1$ and $\sigma \models \phi_2$.
}

An LTL formula $\phi$ is called a \emph{safety property} iff it can be violated by a finite prefix of a computation, or a \emph{liveness property} iff it can only be violated by an infinite computation \cite{BaierKatoenBook}.  For a process $P$ and LTL formula $\phi$, we write $P \models \phi$ iff, for every computation $\sigma$ of $P$, $\sigma \models \phi$.  For convenience, we naturally elevate our notation for satisfaction on computations to satisfaction on runs\techreport{, that is, for a run $r$ of a process $P$ inducing a computation $\sigma$, we write $r \models \phi$ and say ``$r$ satisfies $\phi$" iff $\sigma \models \phi$, or write $r \centernot{\models} \phi$ and say ``$r$ violates $\phi$" iff $\sigma \centernot{\models} \phi$.}{.}
% \techreport{}{For a more thorough explanation see \cite{von2020automated}.}

%% file: sections/Section3AttackerSynthesisProblems.tex
We want to synthesize attackers automatically. 
Intuitively, an attacker is a process that, when composed with the system, violates some property.  
There are different types of attackers, depending on what it means to violate a property (in some cases? in all cases?), as well as on the system topology (threat model).
Next, we define the threat model and attacker concepts formally, followed by the problems considered in this paper.

\subsection{Threat Models}

A \emph{threat model} or \emph{attacker model} prosaically captures the goals and capabilities of an attacker with respect to some victim and environment.  
\techreport{Algebraically, it is difficult to capture the attacker goals and capabilities without also capturing the victim and the environment, so our abstract threat model includes all of the above.}{}
Our threat model captures: how many attacker components there are; how they communicate with each other and with the rest of the system\techreport{: what messages they can intercept, transmit, etc;}{;} and the attacker goals.
\techreport{We formalize the concept of a threat model in what follows.}{}

\begin{definition}[Input-Output Interface]
	An \emph{input-output interface} is a tuple $(I, O)$ such that $I \cap O = \emptyset$ and $I \cup O \neq \emptyset$.  The \emph{class} of an input-output interface $(I, O)$, denoted $\mathcal{C}(I, O)$, is the set of processes with inputs $I$ and outputs $O$.  Likewise, $\mathcal{C}(P)$ denotes the input-output interface the process $P$ belongs to.  (e.g. Fig.~\ref{exampleAttackers}) 
\end{definition}

\begin{sloppypar}
\begin{definition}[Threat Model]
A \emph{threat model} is a tuple $(P, (Q_i)_{i = 0}^m, \phi)$ where $P,Q_0,...,Q_m$ are processes,
each process $Q_i$ has no atomic propositions (i.e., its set of atomic propositions is empty),
and $\phi$ is an LTL formula such that $P \parallel Q_0 \parallel ... \parallel Q_m \models \phi$.
We also require that the system $P \parallel Q_0 \parallel ... \parallel Q_m$ satisfies the formula $\phi$ in a non-trivial manner, that is, that $P \parallel Q_0 \parallel ... \parallel Q_m$ has at least one infinite run.
\end{definition}
\end{sloppypar}
In a threat model,
the process $P$ is called the \emph{target process}, and the processes $Q_i$ are called \emph{vulnerable processes}.  The goal of the adversary is to modify the vulnerable processes $Q_i$ so that composition with the target process violates $\phi$.  (We assume that prior to the attack, the protocol behaves correctly, i.e., it satisfies $\phi$.) \techreport{For example, in $\textsc{TM}_1$ of Fig.~\ref{exampleThreatModels}, the target process is Alice composed with Bob, and the vulnerable processes are Oscar and Trudy, while in $\textsc{TM}_5$, the target process is the composition of Jacob, Simon, Sophia, and Juan, and the only vulnerable process is Isabelle.}{See Fig.~\ref{exampleThreatModels}.}

\techreport{
	\begin{figure}[H]
}{
	\begin{figure}
}
\centering
\begin{adjustbox}{max totalsize={.99\textwidth}{.99\textheight},center}
\begin{tikzpicture}
% BOTTOM LEFT
\draw[draw=black] (-0.5,-0.5) rectangle ++(6.5, 2.5);
\draw[draw=black] (0, 0) rectangle ++(1.5,1.5);
\node[] (AliceBL) at (0.75,0.75) {\large Alice};
\draw[draw=black,dashed] (2, 0) rectangle ++(1.5,1.5);
\node[] (MalloryBL) at (2.75,0.75) {\large Mallory};
\draw[draw=black] (4, 0) rectangle ++(1.5,1.5);
\node[] (BobBL) at (4.75,0.75) {\large Bob};
\draw[straight] (1.5,0.5) to (2,0.5);
\draw[straight] (2,1)     to (1.5,1);
\draw[straight] (3.5,0.5) to (4,0.5);
\draw[straight] (4,1)     to (3.5,1);
\node[draw,rectangle,fill=white,double] (BLlabel) at (2.8,2.1) {\large $\text{TM}_3 = (\text{Alice} \parallel \text{Bob}, (\text{Mallory}), \phi_3)$};
% BOTTOM RIGHT
\draw[draw=black] (6.5,-0.5) rectangle ++(6.5, 2.5);
\draw[draw=black] (7, 0) rectangle ++(1.5,1.5);
\node[] (AliceBR) at (7.75,0.75) {\large Alice};
\draw[draw=black,dashed] (9, 0) rectangle ++(1.5,0.75);
\node[] (EveBR) at (9.75,0.4) {\large Eve};
\draw[draw=black] (11, 0) rectangle ++(1.5,1.5);
\node[] (MarkBR) at (11.75,0.75) {\large Mark};
\draw[straight] (8.5,1)  to (11,1);
\draw[straight] (11,0.5) to (10.5,0.5);
\draw[straight] (9,0.5) to (8.5,0.5);
\node[draw,rectangle,fill=white,double] (BRlabel) at (9.8,2.1) {\large $\text{TM}_4 = (\text{Alice} \parallel \text{Mark}, (\text{Eve}), \phi_4)$};
% TOP RIGHT
\draw[draw=black] (6.5,2.5) rectangle ++(6.5, 2.5);
\draw[draw=black] (7, 3) rectangle ++(1.5,1.5);
\node[] (AliceTR) at (7.75,3.75) {\large Alice};
\draw[draw=black,dashed] (11, 3) rectangle ++(1.5,1.5);
\node[] (BobTR) at (11.75,3.75) {\large Oscar};
\draw[straight] (8.5,4.2) to (11,4.2);
\draw[straight] (11,3.3) to (8.5,3.3);  
\node[draw,rectangle,fill=white,double] (TRlabel) at (9.8,5.1) {\large $\text{TM}_2 = (\text{Alice}, (\text{Oscar}), \phi_2)$};
% TOP LEFT
\draw[draw=black] (-0.5,2.5) rectangle ++(6.5, 2.5);
\draw[draw=black] (0, 3) rectangle ++(1.5,1.5);
\node[] (AliceTL) at (0.75,3.75) {\large Alice};
\draw[draw=black,dashed] (2, 3) rectangle ++(1.5,0.65);
\node[] (TrudyTL) at (2.75,3.3) {\large Trudy};
\draw[draw=black,dashed] (2, 3.86) rectangle ++(1.5,0.65);
\node[] (OscarTL) at (2.75,4.16) {\large Oscar};
\draw[draw=black] (4, 3) rectangle ++(1.5,1.5);
\node[] (BobTL) at (4.75,3.75) {\large Bob};
\draw[straight] (1.5,3.3) to (2,3.3);
\draw[straight] (3.5,3.3) to (4,3.3);
\draw[straight] (2,4.2)   to (1.5,4.2);
\draw[straight] (4,4.2)   to (3.5,4.2);
\node[draw,rectangle,fill=white,double] (TLlabel) at (2.8,5.1) {\large $\text{TM}_1 = (\text{Alice} \parallel \text{Bob}, (\text{Oscar}, \text{Trudy}), \phi_1)$};

\draw[draw=black] (-0.5 + 14,-6 + 5.5) rectangle ++(9.5, 5.5);

% SIMON (bottom left)
\draw[draw=black] (0 + 14, -5.5 + 5.5) rectangle ++(1.5,1.5);
\node[] (Simon5) at (0+0.75 + 14, -5.5+0.75 + 5.5) {\large Simon};

% JACOB (top left)
\draw[draw=black] (0 + 14, -3 + 5.5) rectangle ++(1.5,1.5);
\node[] (Jacob5) at (0+0.75 + 14,-3+0.75 + 5.5) {\large Jacob};

% JUAN (bottom right)
\draw[draw=black] (11-4 + 14, -5.5 + 5.5) rectangle ++(1.5,1.5);
\node[] (Juan5) at (11+0.75-4 + 14,-5.5+0.75 + 5.5) {\large Juan};

% SOPHIA (top right)
\draw[draw=black] (11-4 + 14, -3 + 5.5) rectangle ++(1.5,1.5);
\node[] (Sophia5) at (11+0.75-4 + 14,-3+0.75 + 5.5) {\large Sophia};

% ISABELLE (middle)
\draw[draw=black,dashed] (5-2 + 14, -4.5 + 5.5) rectangle ++(2.5, 2);
\node[] (Isabelle) at (6.25-2 + 14, -3.5 + 5.5) {\large Isabelle};

% ARROWS
% SIMON ---> JUAN
\draw[straight] (1.5 + 14, -5.5+0.2 + 5.5) to (11-4 + 14, -5.5+0.2 + 5.5);
% ISABELLE ---> SIMON
\draw[straight] (5-2 + 14, -4.5 + 0.15 + 5.5) to (1.5 + 14, -5.5 + 0.4 + 5.5);
% SIMON ---> ISABELLE
\draw[straight] (1.5 + 14, -5.5 + 1.5 - 0.4 + 5.5) to (5-2 + 14, -4.5 + 0.15 + 0.7 + 5.5);

\draw[straight] (5-2 + 14, -4.5 + 2 - 0.15 + 5.5) to (1.5 + 14, -3 + 1.5 - 0.4 + 5.5);
\draw[straight] (1.5 + 14, -3 + 0.4 + 5.5) to (5-2 + 14, -4.5 + 2 - 0.15 - 0.7 + 5.5);

\draw[straight] (5 + 2.5 - 2 + 14, -4.5 + 0.15 + 5.5) to (11 - 4 + 14, -5.5 + 0.4 + 5.5);
\draw[straight] (11 - 4 + 14, -5.5 + 1.5 - 0.4 + 5.5) to (5 + 2.5 - 2 + 14, -4.5 + 0.15 + 0.7 + 5.5);

\draw[straight] (5 + 2.5 - 2 + 14, -4.5 + 2 - 0.15 + 5.5) to (11 - 4 + 14, -3 + 1.5 - 0.4 + 5.5);
\draw[straight] (11 - 4 + 14, -3 + 0.4 + 5.5) to (5 + 2.5 - 2 + 14, -4.5 + 2 - 0.15 - 0.7 + 5.5);

\node[draw,rectangle,fill=white,double] (TL5Label) at (4 + 14.25,-1 + 6)
	{\large $\textsc{TM}_5 = (\text{Jacob} \parallel \text{Simon} \parallel \text{Sophia} \parallel \text{Juan}, (\text{Isabelle}), \phi_5)$};

\end{tikzpicture}
\end{adjustbox}
\caption{Example Threat Models.  The properties $\phi_i$ are not shown.  Solid and dashed boxes are processes; we only assume the adversary can exploit the processes in the dashed boxes.  $\text{TM}_1$ describes a distributed on-path attacker scenario, $\text{TM}_2$ describes an off-path attacker, $\text{TM}_3$ is a classical man-in-the-middle scenario, and $\text{TM}_4$ describes a one-directional man-in-the middle, or, depending on the problem formulation, an eavesdropper.  $\textsc{TM}_5$ is a threat model with a distributed victim where the attacker cannot affect or read messages from Simon to Juan.  Note that a directed edge in a network topology from Node 1 to Node 2 is logically equivalent to the statement that a portion of the outputs of Node 1 are also inputs to Node 2.  In cases where the same packet might be sent to multiple recipients, the sender and recipient can be encoded in a message subscript.  Therefore, the entire network topology is {\em implicit} in the interfaces of the processes in the threat model according to the composition definition.}
\label{exampleThreatModels}
\end{figure}

\subsection{Attackers}

\begin{definition}[Attacker]
Let $\textsc{TM} = (P, (Q_i)_{i = 0}^m, \phi)$ be a threat model.  Then $\vec{A} = (A_i)_{i = 0}^m$ is called a \emph{$\textsc{TM}$-attacker} if $P \parallel A_0 \parallel ... \parallel A_m \centernot{\models} \phi$, and, for all $0 \leq i \leq m$: $A_i$ is a deterministic process; $A_i$ has no atomic propositions, and $A_i \in \mathcal{C}(Q_i)$.
\end{definition}

The existence of a $(P, (Q_i)_{i = 0}^m, \phi)$-attacker means that if an adversary can exploit all the $Q_i$, then the adversary can attack $P$ with respect to $\phi$.
Note that an attacker $\vec{A}$ cannot succeed by blocking the system from having any runs at all.
Indeed, 
$P \parallel A_0 \parallel ... \parallel A_m \centernot{\models} \phi$ implies that
$P \parallel A_0 \parallel ... \parallel A_m$ has at least one infinite run violating $\phi$.

Real-world computer programs implemented in languages like \textsc{C} or \textsc{Java} are called \emph{concrete}, while logical models of those programs implemented as algebraic transition systems such as processes are called \emph{abstract}.  The motivation for synthesizing abstract attackers is ultimately to recover exploitation strategies that actually work against concrete protocols.  So, we should be able to translate an abstract attacker (Fig.~\ref{exampleAttackers}) into a concrete one (Fig.~\ref{fig:EattackerFiniteTM2code}). Determinism guarantees that we can do this.  We also require the attacker and the vulnerable processes to have no atomic propositions, so the attacker cannot ``cheat" by directly changing the truth-hood of the property it aims to violate.

\techreport{
We can define an attacker for many properties at once by conjoining those properties (e.g. $\phi_1 \land \phi_2 \land \phi_3$), or for many processes at once by composing those processes (e.g. $P_1 \parallel P_2 \parallel P_3$).  We therefore do not lose expressibility compared to a definition that explicitly allows many properties or processes.}{}

For a given threat model many attackers may exist.  We want to differentiate attacks that are more effective from attacks that are less effective.  One straightforward comparison is to partition attackers into those that always violate $\phi$, and those that only sometimes violate $\phi$.  We formalize this notion with $\exists$-attackers and $\forall$-attackers.

\begin{definition}[$\exists$-Attacker vs $\forall$-Attacker]
Let $\vec{A}$ be a $(P, (Q_i)_{i = 0}^m, \phi)$-attacker.  Then $\vec{A}$ is a \emph{$\forall$-attacker} if $P \parallel A_0 \parallel ... \parallel A_m \models \neg \phi$.  Otherwise, $\vec{A}$ is an \emph{$\exists$-attacker}.
\end{definition}

A $\forall$-attacker $\vec{A}$ always succeeds, because $P \parallel \vec{A}\models \neg \phi$ means that \emph{every} behavior of $P \parallel \vec{A}$ satisfies $\neg \phi$, that is, \emph{every} behavior of $P \parallel \vec{A}$ violates $\phi$.  Since $P \parallel \vec{A} \centernot{\models} \phi$, there must exist a computation $\sigma$ of $P \parallel \vec{A}$ such that $\sigma \models \neg \phi$, so, a $\forall$-attacker cannot succeed by blocking.  An $\exists$-attacker is any attacker that is not a $\forall$-attacker, and every attacker succeeds in at least one computation, so an $\exists$-attacker sometimes succeeds, and sometimes does not.  In most real-world systems, infinite attacks are impossible, implausible, or just uninteresting.
To avoid such attacks, we define an attacker that produces finite-length sequences of adversarial
behavior, and then ``recovers'', meaning that it behaves like the vulnerable process it replaced (see Fig.~\ref{acyclicAttacker}).

\begin{definition}[Attacker with Recovery]
Let $\vec{A}$ be a $(P, (Q_i)_{i = 0}^m, \phi)$-attacker.  If, for each $0 \leq i \leq m$, the attacker component $A_i$ consists of a finite directed acyclic graph (DAG) ending in the initial state of the vulnerable process $Q_i$, followed by all of the vulnerable process $Q_i$, then we say the attacker $\vec{A}$ is an \emph{attacker with recovery}.  We refer to the $Q_i$ postfix of each $A_i$ as its \emph{recovery}
\techreport{, as in Fig.~\ref{exampleAttackers}, $A_3$.}{.}  
\techreport{Intuitively, the vulnerable process is appended to the attacker.}{}
\label{rattackerDefinition}
\end{definition}

Note that researchers sometimes use ``recovery" to mean when a system undoes the damage caused by an attack.  We use the word differently, to mean when the property $\phi$ remains violated even under \emph{modus operandi} subsequent to attack termination.

\techreport{
	\begin{figure}[H]
}{
	\begin{figure}
}
\centering
\begin{adjustbox}{max totalsize={.9\textwidth}{.89\textheight},center}
\begin{tikzpicture}
% P
\node[] (emptyP) {};
\node[draw,circle] (p0) [right=of emptyP] {$p_0 : \{ \texttt{OK} \}$};
\node[draw,circle] (p1) [right=of p0] {$p_1 : \emptyset $};
\draw[straight] (emptyP) to (p0);
\draw[looped] (p0) to[out=north,in=north east,above,looseness=5] node {$a?$} (p0);
\draw[looped] (p0) to[out=south,in=south west,below,looseness=5] node {$c?$} (p0);
\draw[straight] (p0) to[above] node {$b?$} (p1);
\draw[looped] (p0) to[out=south east,in=south west,below] node {$c?$} (p1);
\draw[looped] (p1) to[out=north west,in=north east,above,looseness=5] node {$a?,b?,c?$} (p1);
% Q
\node[draw,circle] (q0) [right=of p1] {$q_0$};
\node[] (emptyQ) [below=of q0] {};
\draw[straight] (emptyQ) to (q0);
\draw[looped] (q0) to[out=north west,in=north east,above,looseness=5] node {$a!$} (q0);
% A1: A-attack with recovery
\node[draw,circle] (a0) [right=of q0] {$a_0^1$};
\node[] (emptyA1) [below=of a0] {};
\draw[straight] (emptyA1) to (a0);
\draw[looped] (a0) to[out=north west,in=north east,looseness=5,above] node {$b!$} (a0);
% A2: E-attack without recovery
\node[draw,circle] (w0) [right=of a0] {$a_0^2$};
\node[] (emptyA2) [below=of w0] {};
\draw[straight] (emptyA2) to (w0);
\draw[looped] (w0) to[out=north west,in=north east,looseness=5,above] node {$c!$} (w0);

\draw[dashed] (-0.8+12,-2.8) rectangle (3.3+12,1.8);
\draw[dashed] (-0.5+12,-1.5) rectangle (3+12,1);
\draw[dashed,fill=white] (0.1+12,-2.4) rectangle (2.2+12,-0.7);

\node[] (A3 ) at (2.3+12,2)   {$A_3$};
\node[] (DAG) at (2.3+12,1.2) {DAG};
\node[] (Q3 ) at (2.4+12,-0.95) {$Q$};

\node[] (emptyA1) at (12,0) {};
\node[draw,circle,fill=white] (a0) [right=of emptyA1] {$a_0^3$};
\node[draw,circle,fill=white] (a3) [below=of a0] {$q_0$};
\draw[straight] (emptyA1) to (a0);
\draw[straight] (a0) to[right] node {$b!$} (a3);
\draw[looped] (a3) to[out=north west,in=south west,left,looseness=5] node {$a!$} (a3);
\end{tikzpicture}
\end{adjustbox}
\caption{From left to right: processes $P$, $Q$, $A_1$, $A_2$, $A_3$.  Let $\phi = \G\,\texttt{OK}$, and let the interface of $Q$ be $\mathcal{C}(Q) = (\emptyset, \{ a, b, c\})$.  Then $P \parallel Q \models \phi$.  $A_1$ and $A_2$ are both deterministic and have no input states.  Let $\mathcal{C}(A_1) = \mathcal{C}(A_2) = \mathcal{C}(Q)$.  Then, $A_1$ and $A_2$ are both $(P, (Q), \phi)$-attackers. $A_1$ is a $\forall$-attacker, and $A_2$ is an $\exists$-attacker.  $A_3$ is a $\forall$-attacker with recovery consisting of a DAG starting at $a_0^3$ and ending at the initial state $q_0$ of $Q$, plus all of $Q$, namely the recovery.}
\label{exampleAttackers}
\techreport{}{\vspace{-1.5cm}}
\end{figure}

\techreport{
	\begin{figure}[H]
}{
	\begin{figure}
}
\centering
\begin{adjustbox}{max totalsize={.7\textwidth}{.7\textheight},center}
\begin{tikzpicture}
\node[] (Ai) at (5,2.9) {$A_i$};
\draw[draw=gray,dashed] (-1.2,-2.5) rectangle (9.8,2.7);

\node[] (DAG) at (3,2.5) {DAG};
\draw[draw=gray,dashed] (-1,-2.3) rectangle (7,2.3);

\path[draw,use Hobby shortcut,closed=true,fill=black!25]
(7,-1) .. (7.1,1) .. (8,2) .. (9.5,1.4) .. (8.8,-1) .. (7.3,-2);
\node[] (Qi) at (8.3,0) {$Q_i$};

\node[] (empty) at (-1,0) {};
\node[draw,circle] (a0) at (0,0) {$a_0^i$};
\draw[straight] (empty) to (a0);
\node[draw,circle] (a1) at (2,1 ) {$a_1^i$};
\node[draw,circle] (a2) at (2,0 ) {$a_2^i$};
\node[draw,circle] (a3) at (2,-1) {$a_3^i$};
\draw[straight] (a0) to[above] node {$x_0$} (a1);
\draw[straight] (a0) to[above] node {$x_1$} (a2);
\draw[straight] (a0) to[above] node {$x_2$} (a3);
\draw[straight] (a1) to[right] node {$x_3$} (a2);
\node[] (m0) at (4,2 ) {...};
\node[] (m1) at (4,1 ) {...};
\node[] (m2) at (4,0 ) {...};
\node[] (m3) at (4,-1) {...};
\node[] (m4) at (4,-2) {...};
\draw[straight] (a3) to[above] node {$x_4$} (m2);
\draw[straight] (a1) to[above] node {$x_5$} (m0);
\draw[straight] (a1) to[above] node {$x_6$} (m1);
\draw[straight] (a2) to[above] node {$x_7$} (m2);
\draw[straight] (a3) to[above] node {$x_8$} (m3);
\draw[straight] (a3) to[above] node {$x_9$} (m4);
\draw[straight] (a0) to[out=south east,in=west,below,looseness=1] node {$x_{10}$} (m4);
\node[draw,circle,fill=white] (q0) at (7,0) {$q_0^i$};
\draw[straight] (m0) to[above] node {$x_k$} (q0);
\draw[straight] (m1) to[above] node {$x_{k+1}$} (q0);
\draw[straight] (m2) to[above] node {$x_{k+2}$} (q0);
\draw[straight] (m3) to[above] node {$x_{k+3}$} (q0);
\draw[straight] (m4) to[left] node {$x_{k+4}$} (q0);
\end{tikzpicture}
\end{adjustbox}
\caption{Suppose $\vec{A} = (A_i)_{i = 0}^m$ is attacker with recovery for $\textsc{TM} = (P, (A_i)_{i = 0}^m, \phi)$.  Further suppose $A_i$ has initial state $a_0^i$, and $Q_i$ has initial state $q_0^i$.  Then $A_i$ should consist of a DAG starting at $a_0^i$ and ending at $q_0^i$, plus all of $Q_i$, called the \emph{recovery}, indicated by the shaded blob.  Note that if some $Q_i$ is non-deterministic, then there can be no attacker with recovery, because $Q_i$ is a subprocess of $A_i$, and all the $A_i$s must be deterministic in order for $\vec{A}$ to be an attacker.}
\label{acyclicAttacker}
\techreport{}{\vspace{-1.2cm}}
\end{figure}

\subsection{Attacker Synthesis Problems}

Each type of attacker - $\exists$ versus $\forall$, with recovery versus without - naturally induces a synthesis problem.

\begin{problem}[$\exists$-Attacker Synthesis Problem ($\exists$ASP)]
Given a threat model \textsc{TM}, find a \textsc{TM}-attacker, if one exists; otherwise state that none exists.
\label{EASP}
\end{problem}
\begin{problem}[Recovery $\exists$-Attacker Synthesis Problem (R-$\exists$ASP)]
Given a threat model \textsc{TM}, find a \textsc{TM}-attacker with recovery, if one exists; otherwise state that none exists.
\label{rEASP} % changed from fEASP
\end{problem}
We defined $\exists$ and $\forall$-attackers to be disjoint, but, 
if the goal is to find an $\exists$-attacker, 
then surely a $\forall$-attacker is acceptable too; we therefore 
did not restrict the $\exists$-problems to only $\exists$-attackers.
Next we define the two $\forall$-problems, which remain for future work.
\begin{problem}[$\forall$-Attacker Synthesis Problem ($\forall$ASP)]
Given a threat model \textsc{TM}, find a \textsc{TM}-$\forall$-attacker, if one exists; otherwise state that none exists.
\label{AASP}
\end{problem}
\begin{problem}[Recovery $\forall$-Attacker Synthesis Problem (R-$\forall$ASP)]
Given a threat model \textsc{TM}, find a \textsc{TM}-$\forall$-attacker with recovery, if one exists; otherwise state that none exists.
\label{rAASP}
\end{problem}

%% file: sections/Section4Solutions.tex
We present solutions $\exists$ASP and R-$\exists$ASP for any number of attackers, and for both safety and liveness properties.  Our success criteria are soundness and completeness.  
Both solutions are polynomial in the product of the size of~$P$ and the sizes of the interfaces of the~$Q_i$s, and exponential in the size of the property~$\phi$~\cite{IEEECS86}.
For real-world performance, see Section~\ref{sec:tcp}.

We reduce $\exists$ASP and R-$\exists$ASP to model-checking.  
The idea is to replace the vulnerable processes $Q_i$ with appropriate ``gadgets'',
then ask a model-checker whether the system violates a certain property. 
We prove that existence of a violation (a counterexample) \rem{is equivalent
to}\add{implies} existence of an attacker, \add{and for the $\exists$ASP, that the implication goes both ways.} 
\add{Then,}\rem{and} we show how to transform the counterexample into
an attacker. The gadgets and the LTL formula are different, depending on whether we seek 
attackers without or with recovery.

\subsection{Gadgetry}

\techreport{We begin by defining \emph{lassos} and \emph{bad prefixes}.}{}
A computation $\sigma$ is a \emph{lasso} if it equals a finite word $\alpha$, then infinite repetition of a finite word $\beta$, i.e., $\sigma = \alpha \cdot \beta^{\omega}$.  A prefix $\alpha$ of a computation $\sigma$ is called a \emph{bad prefix} for $P$ and $\phi$ if $P$ has $\geq 1$ runs inducing computations starting with $\alpha$, and every computation starting with $\alpha$ violates $\phi$.  We naturally elevate the terms \emph{lasso} and \emph{bad prefix} to runs and their prefixes.  We assume a \emph{model checker}: a procedure \textsc{MC}$(P, \phi)$ that takes as input a process $P$ and property $\phi$, and returns $\emptyset$ if $P \models \phi$, or one or more violating lasso runs or bad prefixes of runs for $P$ and $\phi$, otherwise \cite{BaierKatoenBook}.

Attackers cannot have atomic propositions.  So, the only way for $\vec{A}$ to \emph{attack} $\textsc{TM}$ is by sending and receiving messages, hence the space of attacks is within the space of labeled transition sequences.  The \emph{Daisy Process} nondeterministically exhausts the space of input and output events of a vulnerable process.

\begin{definition}[Daisy Process]
	Let $Q = \langle \emptyset, I, O, S, s_0, T, L \rangle$ be a process with no atomic propositions.  Then the \emph{daisy} of $Q$, denoted $\textsc{Daisy}(Q)$, is the process defined below, where $L' : \{ d_0 \} \to \{ \emptyset \}$ is the map such that $L'(d_0)=\emptyset$.

\begin{equation}
\textsc{Daisy}(Q) = \langle \emptyset, I, O, \{ d_0 \}, d_0, \{ (d_0, w, d_0) \mid w \in I \cup O \}, L' \rangle
\end{equation}
\techreport{
	\begin{figure}[H]
	\centering
	\begin{tikzpicture}
	\node[] (empty) {};
	\node[draw,circle] (d0) [right=of empty] {$d_0$};
	\draw[straight] (empty) to (d0);
	\draw[looped] (d0) to[out=north west,in=north east,looseness=5,above] node {$i? \text{ for } i \in I$} (d0);
	\draw[looped] (d0) to[out=south east,in=south west,looseness=5,below] node {$o! \text{ for } o \in O$} (d0);
	\end{tikzpicture}
	\caption{$\textsc{Daisy}(Q)$ has transitions $d_0 \xrightarrow[]{i?} d_0$ and $d_0 \xrightarrow[]{o!} d_0$ for each $i \in I$ and $o \in O$.  So if $Q$ has a run $s_0 \xrightarrow[]{m_0} s_1 \xrightarrow[]{m_1} s_2 \xrightarrow[]{m_2} ...$, then $\textsc{Daisy}(Q)$ must have an I/O-equivalent run $d_0 \xrightarrow[]{m_0} d_0 \xrightarrow[]{m_1} d_0 \xrightarrow[]{m_2} ...$.  In other words, $\textsc{Daisy}(Q)$ can do everything $Q$ can do, and more.}
	\end{figure}
}{}
\end{definition}

Next, we define a \emph{Daisy with Recovery}.  This gadget is an \emph{abstract process}, i.e., a generalized process with a non-empty set of initial states $S_0 \subseteq S$.  Composition and LTL semantics for abstract processes are naturally defined.  We implicitly transform processes to abstract processes by wrapping the initial state in a set.

\begin{definition}[Daisy with Recovery]
Given a process $Q_i = \langle \emptyset, I, O, S, s_0, T, L \rangle$, the \emph{daisy with recovery} of $Q_i$, denoted $\textsc{RDaisy}(Q_i)$, is the abstract process $\textsc{RDaisy}(Q_i) = \langle \text{AP}, I, O, S', S_0, T', L' \rangle$, with atomic propositions $\text{AP} = \{ \texttt{recover}_i \}$, states $S' = S \cup \{ d_0 \}$, initial states $S_0 = \{ s_0, d_0 \}$, transitions \(T' = T \cup \{ (d_0, x, w_0) \mid x \in I \cup O, w_0 \in S_0 \}\), and labeling function $L' : S' \to 2^\text{AP}$ that takes $s_0$ to $\{ \texttt{recover}_i \}$ and other states to $\emptyset$.  (We reserve the symbols $\texttt{recover}_0, ...$ for use in daisies with recovery, so they cannot be sub-formulae of the property in any threat model.)
\end{definition}
\techreport{
The daisy with recovery gadget is illustrated in Fig.~\ref{fig_daisywithrecovery}.

\begin{figure}[H]

\centering
\techreport{}{\begin{adjustbox}{max totalsize={.5\textwidth}{.4\textheight},center}}
\begin{tikzpicture}
\path[draw,use Hobby shortcut,closed=true,fill=black!25]
(7-2.9,-1+2) .. (7.1-2.9,1+2) .. (8-2.9,2+2) .. (9.5-2.9,1.4+2) .. (8.8-2.9,-1+2) .. (7.3-2.9,-2+2);
\node[] (Qi) at (8.3-3,0+3) {$Q_i$};

\node[] (empty) {};
\node[draw,circle] (d0) [above=of empty] {$d_0 : \emptyset$};
\draw[straight] (empty) to (d0);
\draw[looped] (d0) to[out=north west,in=north east,looseness=5,above] node {$i? \text{ for } i \in I$} (d0);
\draw[looped] (d0) to[out=north west,in=south west,looseness=5,left] node {$o! \text{ for } o \in O$} (d0);
\node[] (empty2) [right=of d0] {};
\node[draw,circle,fill=white] (s0) [right=of empty2] {$s_0 : \{ \texttt{recover}_i \}$};
\node[] (empty3) [below=of s0] {};
\draw[straight] (empty3) to (s0);
\draw[straight] (d0) to[above] node[align=center,above] {$i? \text{ for } i \in I$\\$o! \text{ for } o \in O$} (s0);
\end{tikzpicture}
\techreport{}{\end{adjustbox}}
\caption{$\textsc{RDaisy}(Q_i)$ is the result of (1) directly connecting $\textsc{Daisy}(Q_i)$ to the initial state of $Q_i$ in every possible way, (2) allowing both the initial state of the daisy and the initial state of $Q_i$ to be initial states of the abstract process, and (3) adding a single atomic proposition $\texttt{recover}_i$ to the initial state of $Q_i$.}
  \label{fig_daisywithrecovery}
\end{figure}
}{}

\subsection{Solution to $\exists$ASP \label{EASPSolution}}

Let $\textsc{TM} = (P, (Q_i)_{i = 0}^m, \phi)$ be a threat model.  
Our goal is to find an attacker for $\textsc{TM}$, if one exists, or state that none exists, otherwise.
First, we check whether the system
$P \parallel \textsc{Daisy}(Q_0) \parallel ... \parallel \textsc{Daisy}(Q_m)$ satisfies $\phi$.
If it does, then no attacker exists, as the daisy processes encompass any possible attacker
behavior.  
Define a set $R$ returned by the model-checker \textsc{MC}:
\begin{equation}
R = \textsc{MC}(P \parallel \textsc{Daisy}(Q_0) \parallel ... \parallel \textsc{Daisy}(Q_m), \phi)
\label{equationR}
\end{equation}
If $R = \emptyset$ then no attacker exists. 
\techreport{
	\begin{tcolorbox}[colback=lightgreen!20,colframe=lightgreen,sharp corners,title=Solution to $\exists$ASP (Pseudocode),coltitle=black]
	
	\texttt{1.} \(R \gets \textsc{MC}(
		P \parallel \textsc{Daisy}(Q_0) \parallel 
			... \parallel \textsc{Daisy}(Q_m), 
		\phi)\)
	
	\texttt{2.} \(\textit{If }(R = \emptyset) 
		\textit{ then } \textbf{return} \, \text{``no \textsc{TM}-attacker exists"}\)
	\end{tcolorbox}
}{}
On the other hand, if the system violates $\phi$, then we can transform a violating run into a set of attacker processes by projecting it onto the corresponding
interfaces.  Choose a violating run or bad prefix $r \in R$ arbitrarily.  Either $r = \alpha$ is some finite bad prefix, or $r = \alpha \cdot \beta^{\omega}$ is a violating lasso.  For each $0 \leq i \leq m$, let $\alpha_i$ be the projection of $\alpha$ onto the process $\textsc{Daisy}(Q_i)$.  That is, let $\alpha_i = []$; then for each $(\vec{s}, x, \vec{s}')$ in $\alpha$, if $x$ is an input or an output of $Q_i$, and $q, q'$ are the states $\textsc{Daisy}(Q_i)$ embodies in $\vec{s}, \vec{s}'$, add $(q, x, q')$ to $\alpha_i$.
\techreport{
	\begin{tcolorbox}[colback=lightgreen!20,colframe=lightgreen,sharp corners,]

	\texttt{3.} \textbf{Choose} $r \in R$ arbitrarily.

	\texttt{4.} Either $r = \alpha \cdot \beta^{\omega}$ is a lasso, or $r = \alpha$ is
	a bad prefix.

	\texttt{5.} $(\alpha_i)_{i = 0}^m \gets ([])_{i = 0}^m$

	\texttt{6.} $\vec{A} \gets (\perp)_{i = 0}^m$

	\texttt{7.} \textit{For} $0 \leq i \leq m$ \textit{do}

	\texttt{7.1.} \textemdash\textemdash \, \textit{For} $0 \leq j < \abs{\alpha}$ \textit{do}

	\texttt{7.1.1.} \textemdash\textemdash\textemdash\textemdash \, $(\vec{s}, x, \vec{s}') \gets \alpha[j]$

	\texttt{7.1.2.} \textemdash\textemdash\textemdash\textemdash \, \textit{If }$(x \in I_i \cup O_i) \textit{ then } \alpha_i.\texttt{append}((\vec{s}[i], x, \vec{s}'[i]))$

	\end{tcolorbox}
}{}
For each $\alpha_i$, create an incomplete process $A_i^{\alpha}$ with a new state $s_{j + 1}^{\alpha}$ and 
transition $s_j^{\alpha} \xrightarrow[]{z} s_{j + 1}^{\alpha}$ 
for each $\alpha_i[j] = (d_0^i, z, d_0^i)$ for $0 \leq j < \abs{\alpha_i}$.  
If $r = \alpha \cdot \beta^{\omega}$ is a lasso, then for each $0 \leq i \leq m$, define $A_i^{\beta}$ 
from $\beta_i$ in the same way that we defined $A_i^{\alpha}$ from $\alpha_i$; 
let $A_i'$ be the result of merging the first and last states of $A_i^{\beta}$ with the last state of $A_i^{\alpha}$.
\techreport{
	\begin{tcolorbox}[colback=lightgreen!20,colframe=lightgreen,sharp corners,]

	\texttt{7.2.} \textemdash\textemdash \, \textbf{Create} a new state $s_0^{\alpha_i}$

	\texttt{7.3.} \textemdash\textemdash \, $S_i \gets \{ s_0^{\alpha_i} \}$

	\texttt{7.4.} \textemdash\textemdash \, $T_i \gets \emptyset$

	\texttt{7.5.} \textemdash\textemdash \, \textit{For} $0 \leq j < \abs{\alpha_i}$ \textit{do}

	\texttt{7.5.1.} \textemdash\textemdash\textemdash\textemdash \, $(\_, z, \_) \gets \alpha_i[j]$

	\texttt{7.5.2.} \textemdash\textemdash\textemdash\textemdash \, \textbf{Create} a new state $s_{j + 1}^{\alpha_i}$

	\texttt{7.5.3.} \textemdash\textemdash\textemdash\textemdash \, $S_i.\texttt{append}(s_{j + 1}^{\alpha_i})$

	\texttt{7.5.4.} \textemdash\textemdash\textemdash\textemdash \, $T_i.\texttt{append}((s_j^{\alpha_i}, z, s_{j + 1}^{\alpha_i}))$

	\texttt{7.6.} \textemdash\textemdash \, Let $L_i : S_i \to 2^{\emptyset}$ denote the function $\lambda s . \emptyset$

	\texttt{7.7.} \textemdash\textemdash \, $A_i^{\alpha} \gets \langle \emptyset, I_i, O_i, S_i, T_i, L_i \rangle$

	\texttt{7.8.} \textemdash\textemdash \, \textit{If }$(r = \alpha \cdot \beta^{\omega} \text{ is a lasso})$\textit{ then }

	\texttt{7.8.1.} \textemdash\textemdash\textemdash\textemdash \, Define $\beta_i$ from $\beta$ in the same way we defined $\alpha_i$ from $\alpha$ in \texttt{7.1}

	\texttt{7.8.2.} \textemdash\textemdash\textemdash\textemdash \, Define $A_i^{\beta} = \langle \emptyset, I_i, O_i, S_i^{\beta}, T_i^{\beta}, L_i^{\beta} \rangle$ from $\beta_i$ as we defined $A_i^{\alpha}$ from $\alpha_i$ in \texttt{7.2-7.6}

	\texttt{7.8.3.} \textemdash\textemdash\textemdash\textemdash \, $A_i \gets$ $A_i^{\alpha}$ glued to $A_i^{\beta}$ where we merge $A_i^{\alpha}$'s final state $s_{\abs{\alpha_i}}^{\alpha_i}$ with $A_i^{\beta}$'s initial state $s_0^{\beta_i}$

	\end{tcolorbox}
}{}
Otherwise, if $r = \alpha$ is a bad prefix, 
let $A_i'$ be the result of adding an input self-loop to the last state of $A_i^{\alpha}$, 
or an output self-loop if $Q_i$ has no inputs.  Either way, $A_i'$ is an incomplete attacker.  
Finally let $A_i$ be the result of making every input state in $A_i'$ input-enabled via self-loops, 
and return the attacker $\vec{A} = (A_i)_{i = 0}^m$.
\techreport{
	\begin{tcolorbox}[colback=lightgreen!20,colframe=lightgreen,sharp corners,]

	\texttt{7.9.} \textemdash\textemdash \, \textit{Else if} $(I_i \neq \emptyset)$ \textit{then }

	\texttt{7.9.1.} \textemdash\textemdash\textemdash\textemdash \, $A_i \gets$ the result of adding a single arbitrary input self-loop to the final state $s_{\abs{\alpha_i}}^{\alpha_i}$ of $A_i^{\alpha}$

	\texttt{7.10.} \textemdash\textemdash \, \textit{Else} $A_i \gets$ the result of adding a single arbitrary output self-loop to the final state $s_{\abs{\alpha_i}}^{\alpha_i}$ of $A_i^{\alpha}$

	\texttt{7.11.} \textemdash\textemdash \, \textit{For} $s \in \text{states}(A_i)$ \textit{do}

	\texttt{7.11.1.} \textemdash\textemdash\textemdash\textemdash \, \textit{If} $s$ is an input state \textit{then} \textit{For} $x \in I_i$ \textit{do} $\text{transitions}(A_i).\texttt{append}((s, x, s))$

	\texttt{8.} \textbf{return} $\vec{A}$

	\end{tcolorbox}
}{}
An illustration of the method is given in Figure~\ref{exampleThreatModelR}.

\begin{figure}[ht]
\begin{mdframed}
\textbf{Threat Model}: $\textsc{TM}' = (P, (Q_0, Q_1), \phi)$, where the processes from left to right are $P$, $Q_0$, and $Q_1$, and where $\phi = \F \G\, l$.  $P$ has inputs \techreport{$k$ and $m$,}{$k, m$,} and output $n$.  $Q_0$ has no inputs, and output $m$.  $Q_1$ has inputs \techreport{$n$ and $h$,}{$n, h$,} and output $k$.  Recall that $P \parallel Q_0 \parallel Q_1 \models \phi$.

\begin{adjustbox}{max totalsize={.8\textwidth}{.7\textheight},center}
\begin{tikzpicture}
% P
\node[]            (empty) at (-3.3, 0) {};
\node[draw,circle] (p0)    at (-2,   0) {$p_0 : \emptyset$};
\node[draw,circle] (p1)    at (0,    0) {$p_1 : \emptyset$};
\node[draw,circle] (p2)    at (1+1,    0) {$p_2 : \emptyset$};
\node[draw,circle] (p3)    at (3+1,    0) {$p_3 : \{ l \}$};
\draw[straight]    (empty) to (p0);
\draw[straight]    (p0)    to[above] node {$k?$} (p1);
\draw[straight]    (p1)    to[above] node {$m?$} (p2);
\draw[straight]    (p2)    to[below] node {$m?$} (p3);
\draw[looped]      (p3)    to[out=north east,in=south east,right,looseness=3] node {$k?$} (p3);
\draw[looped]      (p3)    to[out=south west,in=south east,below] node {$n!$} (p2);
% Q_0
\node[]            (emptyQ) at (6, 0) {};
\node[draw,circle] (q00)    at (6+1,   0) {$q_0^0$};
\draw[straight]    (emptyQ) to (q00);
\draw[looped]      (q00)    to[out=north east,in=north west,looseness=5,above] node {$m!$} (q00);
% Q_1
\node[]            (emptyQ1) at (7.5+1, 0) {};
\node[draw,circle] (q01)     at (8.5+1, 0) {$q_0^1$};
\node[draw,circle] (q11)     at (10+1 , 0) {$q_1^1$};
\draw[straight]    (emptyQ1) to (q01);
\draw[straight]    (q01)     to[above] node {$n?$} (q11);
\draw[looped]      (q01)     to[out=north east,in=north west,looseness=5,above] node {$k!$} (q01);
\draw[looped]      (q11)     to[out=north east,in=north west,looseness=5,above] node {$k!$} (q11);
\end{tikzpicture}
\end{adjustbox}

\textbf{Violating run}: A run $r \in R$ where $R$ is defined as in Equation~\ref{equationR}.

\[r =
\overbrace{
\begin{bmatrix}p_0\\d_0^0\\d_0^1\end{bmatrix}
\xrightarrow[]{k!}
\begin{bmatrix}p_1\\d_0^0\\d_0^1\end{bmatrix}
\xrightarrow[]{m!}}^{\large \alpha}
\overbrace{
\begin{bmatrix}p_2\\d_0^0\\d_0^1\end{bmatrix}
\xrightarrow[]{m!}
\begin{bmatrix}p_3\\d_0^0\\d_0^1\end{bmatrix}
\xrightarrow[]{n!}
\begin{bmatrix}p_2\\d_0^0\\d_0^1\end{bmatrix}
\xrightarrow[]{m!}
\begin{bmatrix}p_3\\d_0^0\\d_0^1\end{bmatrix}
\xrightarrow[]{n!}
...}^{\large \beta^{\omega}} \in R
\]

\textbf{Application of solution:} $r$ is projected and translated into an attacker $\vec{A} = (A_0, A_1)$.

\begin{adjustbox}{max totalsize={.9\textwidth}{.67\textheight},center}
\begin{tikzpicture}
\node[]            (empty0)     at (-1.8, 0) {$\alpha_0 =$};
\node[]            (alpha00)    at (0,  0) {$d_0^0$};
\node[]            (alpha01)    at (2,  0) {$d_0^0$};
\node[]            (empty1)     at (-1.8, -1) {$A_0^{\alpha} =$};
\node[]            (empty2)     at (-1.3, -1) {};
\node[draw,circle] (A0alphaa0) at (0, -1) {$s_0^{\alpha_0}$};
\node[draw,circle] (A0alphaa1) at (2, -1) {$s_1^{\alpha_0}$};
\draw[straight] (alpha00) to[above] node {$m!$} (alpha01);
\draw[straight] (A0alphaa0) to[above] node {$m!$} (A0alphaa1);
\draw[straight] (empty2) to (A0alphaa0);

\node[]            (empty01)    at (4,   0) {$\alpha_1 =$};
\node[]            (alpha10)    at (5.8, 0) {$d_0^1$};
\node[]            (alpha11)    at (7.8, 0) {$d_0^1$};
\node[]            (empty11)    at (4,  -1) {$A_1^{\alpha} =$};
\node[]            (empty21)    at (4.5,-1) {};
\node[draw,circle] (A1alphaa0)  at (5.8, -1) {$s_0^{\alpha_1}$};
\node[draw,circle] (A1alphaa1)  at (7.8, -1) {$s_1^{\alpha_1}$};
\draw[straight] (alpha10) to[above] node {$k!$} (alpha11);
\draw[straight] (A1alphaa0) to[above] node {$k!$} (A1alphaa1);
\draw[straight] (empty21) to (A1alphaa0);

\node[]             (emptyB0)   at (-1.8, -2) {$\beta_0 =$};
\node[]             (beta00)    at (0,    -2) {$d_0^0$};
\node[]             (beta01)    at (2,    -2) {$d_0^0$};
\node[]             (emptyB1)   at (-1.8, -3) {$A_0^{\beta} =$};
\node[]             (emptyB2)   at (-1.3, -3) {};
\node[draw,circle]  (B0alphaa0) at (0,    -3) {$s_0^{\beta_0}$};
\node[draw,circle]  (B1alphaa1) at (2,    -3) {$s_1^{\beta_0}$};
\draw[straight]     (beta00)    to[above] node {$m!$} (beta01);
\draw[straight]     (B0alphaa0) to[above] node {$m!$} (B1alphaa1);
\draw[straight]     (emptyB2)   to (B0alphaa0);

\node[]             (emptyB10)  at (4,    -2) {$\beta_1 = $};
\node[]             (beta10)    at (5.8,  -2) {$d_0^1$};
\node[]             (beta11)    at (7.8,  -2) {$d_0^1$};
\node[]             (emptyB11)  at (4,    -3) {$A_1^{\beta} =$};
\node[]             (emptyB12)  at (4.5,  -3) {};
\node[draw,circle]  (B1alphaa0) at (5.8,  -3) {$s_0^{\beta_1}$};
\node[draw,circle]  (B1alphaa1) at (7.8,  -3) {$s_1^{\beta_1}$};
\draw[straight]     (beta10)    to[above] node {$n?$} (beta11);
\draw[straight]     (B1alphaa0) to[above] node {$n?$} (B1alphaa1);
\draw[straight]     (emptyB12)  to (B1alphaa0);

\node[]             (A0prime)       at (-1.8, -4) {$A_0' =$};
\node[]             (emptyA00prime) at (-1.3, -4) {};
\node[draw,circle]  (A0prime0)      at (0,    -4) {$a_0^{0'}$};
\node[draw,circle]  (A0prime1)      at (2,    -4) {$a_1^{0'}$};
\draw[straight]     (emptyA00prime) to (A0prime0);
\draw[straight]     (A0prime0) to[above] node {$m!$} (A0prime1);
\draw[looped]       (A0prime1) to[out=north east,in=south east,looseness=5,right] node {$m!$} (A0prime1);

\node[]             (A0prime)       at (-1.8, -5) {$A_0 =$};
\node[]             (emptyA00prime) at (-1.3, -5) {};
\node[draw,circle]  (A0prime0)      at (0,    -5) {$a_0^{0}$};
\node[draw,circle]  (A0prime1)      at (2,    -5) {$a_1^{0}$};
\draw[straight]     (emptyA00prime) to (A0prime0);
\draw[straight]     (A0prime0) to[above] node {$m!$} (A0prime1);
\draw[looped]       (A0prime1) to[out=north east,in=south east,looseness=5,right] node {$m!$} (A0prime1);

\node[]             (A1prime)       at (4,    -4) {$A_1' =$};
\node[]             (emptyA10prime) at (4.5,  -4) {};
\node[draw,circle]  (A1prime0)      at (5.8,  -4) {$a_0^{1'}$};
\node[draw,circle]  (A1prime1)      at (7.8,  -4) {$a_1^{1'}$};
\draw[straight]     (emptyA10prime) to (A1prime0);
\draw[straight]     (A1prime0)      to[above] node {$k!$} (A1prime1);
\draw[looped]       (A1prime1)      to[out=north east,in=south east,looseness=5,right] node {$n?$} (A1prime1);

\node[]             (A1)            at (4,    -5) {$A_1 = $};
\node[]             (emptyA10)      at (4.5,  -5) {};
\node[draw,circle]  (A10)           at (5.8,  -5) {$a_0^1$};
\node[draw,circle]  (A11)           at (7.8,  -5) {$a_1^1$};
\draw[straight]     (emptyA10)      to (A10);
\draw[straight]     (A10)           to[above] node {$k!$} (A11);
\draw[looped]       (A11)           to[out=north east,in=south east,looseness=5,right] node {$n?, h?$} (A11);
\end{tikzpicture}
\end{adjustbox}
\end{mdframed}
\caption{Illustration of solution to $\exists$ASP.  Example threat model $\textsc{TM}'$ on top, followed by a violating run in $R$, followed by translation of the run into attacker.}
\label{exampleThreatModelR}
\end{figure}

\begin{theorem}[$\exists$ASP Solution is Sound and Complete]
\label{solution1theorem}
Let $\textsc{TM} = (P, (Q_i)_{i = 0}^m, \phi)$ be a threat model, and define $R$ as in Eqn.~\ref{equationR}.  Then the following hold.  1) $R \neq \emptyset$ iff a $\textsc{TM}$-attacker exists.  2)  If $R \neq \emptyset$, then the procedure above eventually returns a $\textsc{TM}$-attacker.
\end{theorem}

\techreport{
\begin{sketch}
We prove 2) then 1).  Suppose $r \in R$.  Processes are finite and threat models are finitely large, so the procedure eventually terminates.  We need to show the result $\vec{A} = (A_i)_{i = 0}^m$ is a \textsc{TM}-attacker.  Showing the result is deterministic is straightforward, and it should be equally clear that each $A_i$ has the same interface as its respective $Q_i$.  We inductively demonstrate that $P \parallel A_0 \parallel ... \parallel A_m$ has some run $r'$ that is I/O-equivalent to the run $r$ and induces the same computation.  So then $r' \centernot{\models} \phi$, so $\vec{A}$ is a \textsc{TM}-attacker and therefore 2) holds.  

We now turn our attention to 1).  If a \textsc{TM}-attacker $\vec{A}'$ exists, then $P \parallel A_0' \parallel ... \parallel A_m'$ has a run $r$ violating $\phi$.  The daisies can do everything the $Q_i$s can do and more, so the daisies yield some I/O-equivalent run $r'$ violating $\phi$, and so $R \neq \emptyset$.  On the other hand, if $R = \emptyset$ then we can easily prove by way of contradiction that no attacker exists, since attackers, daisies, and vulnerable processes have no atomic propositions, and therefore any violating run of an attacker with $P$ could be translated into an I/O-equivalent run of the daises with $P$ inducing the same computation.  So 1) holds and we are done.
\end{sketch}
}{}

\subsection{Solution to R-$\exists$ASP \label{fEASPSolution}}

Let $\textsc{TM} = (P, (Q_i)_{i = 0}^m, \phi)$ be a threat model as before.  Now our goal is to find a $\textsc{TM}$-attacker \emph{with recovery}, if one exists, or state that none exists, otherwise.
The idea to solve this problem is similar to the idea for finding attackers without recovery, with two differences.
First, the daisy processes are now more complicated, and include recovery to the original $Q_i$ processes.
Second, the formula used in model-checking is not $\phi$, but a more complex formula $\psi$ to ensure that
the attacker eventually recovers, i.e., all the attacker components eventually recover.  We define the property~$\psi$ so that in prose it says ``if all daisies eventually recover, then~$\phi$ holds".  We then define~$R$ like before, except we replace daisies with daisies with recovery, and~$\phi$ with~$\psi$, as defined below.

\begin{equation}
\psi = \big( \bigwedge_{0 \leq i \leq m} \F\, \texttt{recover}_i \big) \implies \phi
\label{equationPsi}
\end{equation}

\begin{equation}
R = \textsc{MC}(P \parallel \textsc{RDaisy}(Q_0) \parallel ... \parallel \textsc{RDaisy}(Q_m), \psi)
\label{equationRrecovery}
\end{equation}

If $R = \emptyset$ then no attacker with recovery exists.
If any $Q_i$ is not deterministic, then likewise no attacker with recovery exists, because our attacker definition requires the attacker to be deterministic, but if $Q_i$ is not and $Q_i \subseteq A_i$ then neither is $A_i$.
\techreport{
	\begin{tcolorbox}[colback=lightgreen!20,colframe=lightgreen,sharp corners,title=Solution to R-$\exists$ASP (Pseudocode),coltitle=black]
	
	\texttt{1.} \(R \gets \textsc{MC}(
		P \parallel \textsc{RDaisy}(Q_0) \parallel 
			... \parallel \textsc{RDaisy}(Q_m), 
		(( \bigwedge_{0 \leq i \leq m} \F\, \texttt{recover}_i) \implies \phi))\)
	
	\texttt{2.} \(\textit{If }(R = \emptyset \text{ or any } Q_i \text{ is not deterministic}) 
		\textit{ then } \textbf{return} \, \text{``no \textsc{TM}-attacker exists"}\)
	\end{tcolorbox}
}{} 

Otherwise, choose a violating run (or bad prefix) $r \in R$ arbitrarily.  We proceed as we did for $\exists$ASP but with three key differences. First, we \rem{define $\alpha_i$ by projecting $\alpha$ onto}\add{use} $\textsc{RDaisy}(Q_i)$ as opposed to $\textsc{Daisy}(Q_i)$.  
\rem{Second, for each $0 \leq i \leq m$, instead of using $A_i^{\beta}$ if $r$ is a lasso, or adding self-loops to the final state if $r$ is a bad prefix, we simply glue $A_i^{\alpha}$ to $Q_i$ by setting the last state of $A_i^{\alpha}$ to be the initial state of $Q_i$.  (The result of gluing is a process; the initial state of $A_i^{\alpha}$ is its only initial state.)}
\add{Second, for each $0 \leq i \leq m$, instead of using $A_i^{\alpha}$ and $A_i^{\beta}$, we compute a program $A_i^{\texttt{rec}}$ which replays the actions taken by $\textsc{RDaisy}(Q_i)$ until $\texttt{recover}_i$ is satisfied.  We glue this program to the recovery $Q_i$.}  
Third, instead of using self-loops to input-enable input states, we use input transitions to the initial state of $Q_i$.  This ensures the pre-recovery portion is a DAG.  Then we return~$\vec{A} = (A_i)_{i = 0}^m$.

	\techreport{
	\begin{tcolorbox}[colback=lightgreen!20,colframe=lightgreen,sharp corners]

	\texttt{3.} Choose $r \in R$ arbitrarily.  For convenience,
\(s_0 \xrightarrow[]{l_0} s_1 \xrightarrow[]{l_1} s_2 \xrightarrow[]{l_2} \dots = r\).

	\texttt{4.} Let $\vec{A}\,:= ()$ be initially empty.

	\texttt{5.} For each $0 \leq i \leq m$:

	\begin{description}

		\item \texttt{5.1.} Let $\textsc{preRecovery}_i$ be the shortest prefix of $r$ ending in 
    a state $s$ satisfying $\texttt{recover}_i \in L(s)$.
    For convenience, $s_0 \xrightarrow[]{l_0} s_1 \xrightarrow[]{l_1} \dots \xrightarrow[]{l_d} s_d = \textsc{preRecovery}_i$.

    \item \texttt{5.2.} Let $A_i^{\texttt{rec}}$ begin as the trivial process with one state $a_0^i$ and no transitions.

    \item \texttt{5.3.} Let $\textsc{NS}$ be a temporary variable initialized to $0$.

    \item \texttt{5.4.} For each $0 \leq i \leq d$, if $l_i$ is in the interface of $Q_i$, 
        then add a new state $a_{\textsc{NS} + 1}$
        and new transition $a_{\textsc{NS}} \xrightarrow[]{l_i} a_{\textsc{NS} + 1}$
        to the process $A_i^{\texttt{rec}}$, and increment $\textsc{NS}$.

    \item \texttt{5.5.} For each input state $a$ of $A_i$,
        let $x?$ be the label on the outgoing input transition $a \xrightarrow[]{x?} a'$
        from $a$; then for each $y \neq x$ in the inputs of $Q_i$,
        add the transition $a \xrightarrow[]{y?} a_{\textsc{NS}}$ to $A_i^{\texttt{rec}}$.

    \item \texttt{5.6.}  Let $A_i$ be initially equal to $A_i^{\texttt{rec}}$.

    \item \texttt{5.7.} Replace the state $a_{\textsc{NS}}$ of $A_i$ 
    										with the initial state $q_0^i$ of $Q_i$, 
    										and add the rest of $Q_i$ to $A_i$.

    \item \texttt{5.8.} Give $A_i$ the same interface as $Q_i$.

    \item \texttt{5.9.} Append $A_i$ to $\vec{A}$.

	\end{description}

	\texttt{6.} Return $\vec{A}$.
	\end{tcolorbox}
	}

\add{
	Next we prove that our solution is sound, and that
	our solution is complete for the subset of attackers with
	recovery which eventually recover, which we refer to as \emph{terminating}.}
\begin{definition}[Terminating Attacker]
    Let $\vec{A}$ be a \textsc{TM}-attacker with recovery.
    If $P \parallel A_0 \parallel \dots \parallel A_m$ has at least one run
    violating $\phi$ in which each attacker component $A_i$ reaches the 
    initial state $q_0^i$ of its
    corresponding recovery~$Q_i$,
    then we say that~$\vec{A}$ is a \emph{Terminating Attacker}.
\end{definition}
\add{In order to prove these results, we need to first prove some smaller Lemmas.
Soundness means that given a threat model~\textsc{TM}$=(P,(Q_i)_{i=0}^m,\phi)$,
    if the solution returns a result~$\vec{A}$,
    then~$\vec{A}$ is a \textsc{TM}-attacker with recovery.
Given some such~\textsc{TM} and~$\vec{A}$,
    we show in Lemma~\ref{lem:attackerSize} 
        that~$\vec{A}$ is the correct size (consisting of $m + 1$ components);
    in Lemma~\ref{lem:attackReqs}
        that~$\vec{A}$ satisfies the syntactic requirements for an attacker;
    and in Lemma~\ref{lem:dag}, that~$\vec{A}$ further satisfies the syntactic
    requirements for an attacker with recovery.
We conclude in Theorem~\ref{thm:sound} by showing that~$\vec{A}$
    also satisfies the semantic requirements for an attacker with recovery.
    (Indeed, $\vec{A}$ is terminating.)
    This suffices to show that our solution is sound.}

\add{Subsequently, we prove that our solution is complete for terminating attackers,
	meaning that our solution returns a result~$\vec{A}$ precisely when 
	a terminating \textsc{TM}-attacker exists.
In order to prove that our solution is complete for terminating attackers, 
	we first show that every terminating \textsc{TM}-attacker
	is discoverable using daisies with recovery, in Lemma~\ref{lem:cole}.
Next, in Theorem~\ref{thm:complete}, we use Lemma~\ref{lem:cole} to show that
 if a terminating attacker exists
then the solution returns some result~$\vec{A}$, 
and if no terminating attacker exists
then the solution states as much and terminates.
This suffices to show that our solution is complete for the terminating subclass
of attackers with recovery.
}

\begin{lemma}[\add{Attacker Size}]
\label{lem:attackerSize}
\add{Given a threat model $\textsc{TM}=(P,(Q_i)_{i=0}^m,\phi)$,
if the solution above returns a tuple $\vec{A}$,
it will have $m+1$ components.}
\end{lemma}
\begin{proof}
\add{If the solution above returns a tuple $\vec{A}$, then it must have progressed
to line~3.
In this case, we know that $R$ is non-empty in line~2.
Then in line~5, the solution iterates $m + 1$ times.
In each iteration, the solution progresses through lines 5.1. through 5.9.
Importantly, each of these lines terminates.
(Line~5.1. terminates because $r \models \F \texttt{recover}_i$;
    for the other lines termination is obvious.)
At line~5.9., the component $A_i$ is appended to $\vec{A}$.
So after~$m + 1$ iterations of~5.1. through~5.9., there will be~$m + 1$
components $A_0, ..., A_m$ in the tuple~$\vec{A}$,
and we are done.}
\end{proof}

\begin{lemma}[\add{Syntactic Requirements for an Attacker}]
\label{lem:attackReqs}
\add{If the solution above yields a tuple $(A_i)_{i=0}^m$,
    then each $A_i$ is a deterministic process
    with no atomic propositions
    and the same inputs and outputs as the corresponding $Q_i$.}
\end{lemma}
\begin{proof}
\add{First we will show that each $A_i$ is deterministic, i.e., that each $A_i$
satisfies determinism conditions~(i) through~(iv).}
\begin{enumerate}[(i)]
    \item \add{Let $T_i$ refer to the transitions of $A_i$.
    Suppose $T_i$ contains some transitions $s \xrightarrow[]{l} a$
    and $s \xrightarrow[]{l} b$.  If $s$ is in recovery, then $a$ and $b$ must
    likewise be in recovery as there is no way to leave recovery.  In this case,
    since the recovery process~$Q_i$ is assumed to be deterministic, it therefore
    follows that $a = b$.  What if $s$ is not in recovery?  Then the 
    transition~$s \xrightarrow[]{l} a$ was either added at line~5.4. or line~5.5.
    For a contradiction assume $a \neq b$.}
    \begin{itemize}
        \item \add{It cannot be the case that both transitions were added at line~5.4.
        because, at the end of line~5.4., the source state $a_{\textsc{NS}}$
        is incremented, violating the assumption that these distinct 
        transitions begin at some state~$s$.}
        \item \add{If the former was added in~5.4. and the latter in~5.5., then 
        the $s$ must be an input-state.
        In this case, by construction the latter could not have the same label
        as the former, so we've arrived at a contradiction.}
        \item \add{If the latter was added in 5.4. and the former in~5.5. then 
        the result is symmetric with the previous case.}
        \item \add{If both were added in~5.5. then they must have different labels,
        contradicting our assumption that they share the label~$l$.}
    \end{itemize}
    \add{So in all cases, assuming these are distinct transitions, we arrive at a 
    contradiction.  Therefore they must not be distinct transitions, so $a = b$.
    We conclude that the transition relation $T_i$ of $A_i$ can be expressed as
    a (possibly partial) function $S \times (I \cup O) \to S$.
    \item Let $s$ be a non-deadlock state in $A_i$.
    If $s$ is in recovery then, as the recovery is deterministic, 
    $s$ is either an input state or an output state, but not both.
    Else $s$ is in the attack component.  Consider this case.
    Clearly at line~5.4., some transition from $s$ to either the next
    attack state or recovery was added, in order to simulate a corresponding
    transition taken by the daisy with recovery in the violating run $r$.
    If $s$ is an input state, then it has one outgoing input transition added at
    line~5.4., and potentially more outgoing input transitions added at 
    line~5.5., but no other outgoing transitions.
    If $s$ is an output state, then it has one outgoing output transition added
    at line~5.4., and no other outgoing transitions.
    Either way, $s$ is either an input state or an output state, but not both.}
    \item \add{Every input state in recovery is input-enabled because we assume the
    recovery is deterministic.  Every input state before recovery is input-enabled
    because of line~5.5.}
    \item \add{Each output state in recovery has only one outgoing transition because
    we assume recovery is deterministic.  Each output state before recovery has
    only one outgoing transition, as transitions are only added (once) in 
    line~5.4. or (never, if the state is an output state) in line~5.5.
    So, every output state in~$A_i$ has only one outgoing transition.}
\end{enumerate}
\add{We conclude that each~$A_i$ is deterministic.}

\add{At no point in the solution do we assign any atomic propositions to any~$A_i$.
Implicitly, each~$A_i^{\text{rec}}$ is initialized without atomic propositions, in line~5.2.
(as ``the trivial process'').
Since \textsc{TM} is a threat model, the $Q_i$s do not have atomic
propositions.
Hence adding~$Q_i$ to~$A_i^{\text{rec}}$ does not add atomic propositions,
so the~$A_i$s have no atomic propositions.}

\add{Every transition in~$A_i$ is added in line~5.4., added in line~5.5., or
a transition in~$Q_i$.  In the first case, the transition is labeled using an input
or an output of~$Q_i$.  In the second case, the transition is labeled using an input of~$Q_i$.
In the third case, the transition is a transition in~$Q_i$, 
    and thus in the interface of~$Q_i$.
So every transition in~$A_i$ is in the interface of~$Q_i$.
We conclude that each~$A_i$ has the same inputs and outputs as the corresponding~$Q_i$.}

\add{
	Overall, we conclude that if the solution above 
	yields a tuple $(A_i)_{i=0}^m$,
    then each $A_i$ is a deterministic process
    with no atomic propositions
    and the same inputs and outputs as the corresponding $Q_i$.
}
\end{proof}

\begin{lemma}[\add{Syntactic Requirements for an Attacker with Recovery}]
\label{lem:dag}
\add{If the solution above yields a tuple $(A_i)_{i=0}^m$,
    then each $A_i$ consists of a finite DAG ending at the recovery~$Q_i$.}
\end{lemma}
\begin{proof}
\add{Suppose the solution above yields a tuple $\vec{A}$$=(A_i)_{i=0}^m$.
Choose a natural number $0 \leq i \leq m$ arbitrarily.
We want to show that $A_i$ consists of a finite DAG ending at the recovery~$Q_i$.
Transitions outside of recovery in $A_i$ are added in lines 5.4. and 5.5.
The number of such transitions added is linear in the product of the 
size of the interface of $Q_i$ and the length of the finite prefix $\textsc{preRecovery}_i$.
Hence, the part of $A_i$ outside recovery is clearly finite.
But is it a DAG?
The transitions added by line 5.4. form a line from $a_0^i, a_1^i, \dots, a_{d - 1}^i, q_0^i$.
This line is totally ordered by the incrementing subscript on the state name, and is bounded
in length by the size of $\textsc{preRecovery}_i$, and ends in recovery because of line~5.7.
Every other pre-recovery transition is added by line 5.5. and points to $q_0^i$.
So the pre-recovery portion is a DAG ending in $q_0^i$, the initial state of $Q_i$.  As $0 \leq i \leq m$ was arbitrary, it follows that the pre-recovery portion of each $A_i$ is a DAG ending in the corresponding $q_0^i$; so the lemma holds and we are done.}
\end{proof}

\begin{theorem}
\label{thm:sound}
\add{The solution above is sound.}
\end{theorem}
\begin{proof}
\add{Suppose given some threat model $\textsc{TM}=(P,(Q_i)_{i=0}^m,\phi)$,
the solution above returns $\vec{A}$.
For convenience, define \textsc{AttackedSystem} and \textsc{DaisySystem} as
we did previously.
We need to show that $\vec{A}$ is, in fact, a \textsc{TM}-attacker
with recovery.
By Lemma~\ref{lem:attackerSize}, $\vec{A}$
    is a tuple $(A_i)_{i=0}^m$
    having $m + 1$ components.
By Lemma~\ref{lem:attackReqs}, each component $A_i$ 
    is a deterministic process with no
atomic propositions and the same inputs and outputs as the corresponding $Q_i$.
And by Lemma~\ref{lem:dag}, each component $A_i$ 
    consists of a finite DAG ending at the
initial state of the corresponding $Q_i$.
All that remains is to show that
    \textsc{AttackedSystem}
    has at least one run
    violating $\phi$.}

\add{Let $r = \textbf{s}_0 \xrightarrow[]{l_0} \textbf{s}_1 \xrightarrow[]{l_1} \dots$ 
be the run of \textsc{DaisySystem} which was selected in line~3. of our solution.
(Presumably we made it to line 3., since the solution returned an attacker~$\vec{A}$.)
We will show that \textsc{AttackedSystem} can simulate the run~$r$ with some run~$r^{(m)}$ that violates~$\phi$.
We will make an inductive argument, albeit, one that terminates after $m + 1$ 
steps.
Let $\zeta(i)$ be the inductive claim:}
    \[\begin{aligned}
    \zeta(i) := &
    \, \, \exists \, r^{(i)} \in \text{runs}\Big(
        P \parallel A_0 
          \parallel \dots 
          \parallel A_i 
          \parallel \textsc{RDaisy}(Q_{i + 1}) 
          \parallel \dots 
          \parallel \textsc{RDaisy}(Q_m)
    \Big) \\
    & \text{ such that }
    r^{(i)} \centernot{\models} \phi \text{ and in $r^{(i)}$, all the $A_i$s eventually recover. }
    \end{aligned}
    \]
\add{In the base case, $i = -1$, and the system is \textsc{DaisySystem}.
Setting $r^{(-1)} = r$, we see that \textsc{DaisySystem}
   satisfies these conditions, 
   thus $\zeta(-1) = \emph{true}$.
We now move on to the inductive step.}

\add{Assume $\zeta(i)$ holds.
If $i = m$, then the entire theorem holds and we are done.
Assume $i < m$ and let $r^{(i)}$ be the run discovered in the previous
iteration of this argument.}

\vspace{\baselineskip}
\noindent \add{\textbf{Case 1:} Suppose that in $r^{(i)}$, 
    the process $\textsc{RDaisy}(Q_{i + 1})$ begins in the 
    recovery state $q_0^{i + 1}$.
Since the daisy portion of $\textsc{RDaisy}(Q_{i + 1})$
is unreachable from the recovery portion $Q_{i + 1}$,
and since $Q_{i + 1}$ uses all of its inputs and outputs,
it immediately follows that we could delete the daisy portion of 
$\textsc{RDaisy}(Q_{i + 1})$
entirely and still find the same run $r^{(i)}$.
For convenience, let $Q_{i + 1}'$ denote the result of this deletion, i.e.,
$Q_{i + 1}'$ is the same as $Q_{i + 1}$
except that its initial state is additionally labeled \texttt{recover}$_{i + 1}$.
Define the following three systems:}
\[
\begin{aligned}
\textsf{S1} := \, & 
P \parallel A_0 
  \parallel \dots 
  \parallel A_i
  \parallel Q_{i + 1}'
  \parallel \textsc{RDaisy}(Q_{i + 2})
  \parallel \dots
  \parallel \textsc{RDaisy}(Q_m) \\
\textsf{S2} := \, & 
P \parallel A_0 
  \parallel \dots 
  \parallel A_i
  \parallel Q_{i + 1}
  \parallel \textsc{RDaisy}(Q_{i + 2})
  \parallel \dots
  \parallel \textsc{RDaisy}(Q_m) \\
\textsf{S3} := \, &
P \underbrace{\parallel A_0 
  \parallel \dots 
  \parallel A_i}_{\text{empty list if }i=-1}
  \parallel A_{i + 1}
  \parallel \textsc{RDaisy}(Q_{i + 2})
  \parallel \dots
  \parallel \textsc{RDaisy}(Q_m)
\end{aligned}
\]
\add{So 
modulo the propositions
    \texttt{recover}$_0$ through \texttt{recover}$_{i + 1}$, 
    $r^{(i)}$ is also a run of \textsf{S1},
        and thus also \textsf{S2}.
But in this case, the solution will return $A_{i + 1} = Q_{i + 1}$,
so it follows that $r^{(i)}$ is 
    (again, modulo the \texttt{recover} propositions)
also a run of~\textsf{S3}.
Call this run~$r^{(i + 1)}$.}
 
\vspace{\baselineskip}
\noindent \add{\textbf{Case 2:}
Suppose that in $r^{(i)}$, the process 
    $\textsc{RDaisy}(Q_{i + 1})$ 
    begins in the 
    initial daisy state $d_0^{i + 1}$.  
    Notice that other than Case~1, this is the only
remaining possibility, since the initial states of~\textsc{RDaisy}$(Q_{i + 1})$
are~$q_0^{i + 1}$ and~$d_0^{i + 1}$.
Let $t_0,...,t_j$ be the sequence of transitions that $\textsc{RDaisy}(Q_{i + 1})$
takes from $d_0^{i + 1}$ in the prefix $\textsc{preRecovery}_{i + 1}$ 
defined at line~5.1.
of our solution.
Observe the following:}
\begin{enumerate}[(I)]
    \item \add{There exists a sequence of consecutive transitions $t_0',...,t_j'$
    in $A_{i + 1}$ such that $t_0'$ begins in $a_0^{i + 1}$, 
    $t_j'$ ends in $q_0^{i + 1}$,
    and each $t_i'$ has the same label as the corresponding $t_i$.
    This holds because of line~5.4. in our solution.}
    \item \add{Suppose $\textsc{RDaisy}(Q_{i + 1})$ stayed still during 
    some transition~$t$
    in the prefix $\textsc{preRecovery}_{i + 1}$.
    Then regardless of its state, $A_{i + 1}$ would likewise stay still during any
    system transition $t'$ having the same label as $t$.
    This holds because in order for $\textsc{RDaisy}(Q_{i + 1})$ to stay still during
    the transition~$t = \textbf{s} \xrightarrow[]{l} s'$, it would need to be the
    case that $l$ does not match any of the self loops or recovery transitions
    stemming from the initial state $d_0^{i + 1}$ 
    of the daisy portion of $\textsc{RDaisy}(Q_{i + 1})$.
    In which case, $l$ is not in the interface of $Q_{i + 1}$, 
        and thus, by Lemma~\ref{lem:attackReqs},
        it is also not in the interface of $A_{i + 1}$.
        So $A_{i + 1}$ would likewise stay still 
        during any system transition with label~$l$.}
\end{enumerate}
\add{Since both these claims hold, 
    the augmented system can simulate the prefix
$\textsc{preRecovery}_{i + 1}$.
But at the end of this prefix (and its simulation)
    $\textsc{RDaisy}(Q_{i + 1})$ and $A_{i + 1}$ have both recovered, 
    at which point they act identically.
    Keeping in mind that processes must use all their inputs and outputs,
        we conclude that the augmented system can simulate the entire run~$r^{(i)}$.
Call the simulating run of the augmented system $r^{(i + 1)}$.}

\add{Since $d_0^{i + 1}$ and the attacker states $a_i^{i + 1}$ lack atomic propositions,
it follows that modulo the propositions \texttt{recover}$_0$ 
    through \texttt{recover}$_{i + 1}$,
    $r^{(i)}$ and $r^{(i + 1)}$ will induce identical traces.
    By construction, the propositions \texttt{recover}$_0$
    through \texttt{recover}$_{i + 1}$
    does not appear in the property~$\phi$,
    thus both of these traces violate~$\phi$.
In both~$r^{(i)}$ and~$r^{(i + 1)}$,
    the processes 
    $P, A_0, \dots, A_i, \textsc{RDaisy}(Q_{i + 2}), \dots, \textsc{RDaisy}(Q_m)$
    progress through the same sequence of transitions
    and in the same order.
Since we assumed these components recovered in $r^{(i)}$ 
    clearly they also recover in $r^{(i + 1)}$.
Can we say something similar with respect to $\textsc{RDaisy}(Q_{i+1})$ and $A_{i+1}$?  Yes: In Case 1, the attacker component $A_{i + 1}$ begins in recovery, while in Case 2, Line~5.4. of our solution assures the
    transitions $t_0'$ through $t_j'$
    are consecutive, line~5.6. assures that~$t_j'$ ends in recovery,
    and we just proved that transitions $t_0'$ through $t_j'$ are taken by $A_{i + 1}$
    in the run $r^{(i + 1)}$.
So in all cases, all the components recover in both runs.
So the simulation requirements hold, i.e., $\zeta(i + 1) = \emph{true}$.}

\add{Since we set $\zeta(-1) = \emph{true}$,
and proved that for $-1 \leq i < m$,
    $\zeta(i)$ implies $\zeta(i + 1)$, we conclude by way of (finite)
    induction 
    that $\vec{A}$ is a \textsc{TM}-attacker with recovery.
This suffices to show our solution is sound, and we are done.}
\end{proof}

\begin{lemma}[\add{Terminating Attackers Induce Runs Discoverable by Daisies with Recovery}]
\label{lem:cole}
\add{Let $\textsc{TM}=(P,(Q_i)_{i=0}^m,\phi)$ be a threat model and $\vec{A}$ a terminating
\textsc{TM}-attacker.  
Let $\psi = (\bigwedge_{0 \leq i \leq m} \F \texttt{recover}_i ) \implies \phi$.
Then there exists a run $r' \in \text{runs}(P \parallel \textsc{RDaisy}(Q_0) \parallel \dots \parallel \textsc{RDaisy}(Q_m))$ such that $r' \centernot{\models} \psi$.}
\end{lemma}
\begin{proof}
\add{For convenience, define \textsc{AttackedSystem} and \textsc{DaisySystem} as follows:}
\[\begin{aligned}
\textsc{AttackedSystem} & = P \parallel A_0 \parallel \dots \parallel A_m \\
\textsc{DaisySystem} & = P \parallel \textsc{RDaisy}(Q_0) \parallel \dots \parallel \textsc{RDaisy}(Q_m)
\end{aligned}\]
\add{Since~$\vec{A}$ is terminating,
    $\textsc{AttackedSystem}$
    has a run~$r$
    in which every~$A_i$ visits the initial state~$q_0^i$ of its recovery~$Q_i$,
    and~$r \centernot{\models} \phi$.
    Let~$r$ be some such run.
    We will show that $\textsc{DaisySystem}$ can ``simulate'' the run~$r$, i.e.,
    that $\textsc{DaisySystem}$
    has some run~$r'$
    such that the only differences between~$r$ and~$r'$ 
    are the names of the states the~$A_i$s or~$\textsc{RDaisy}(Q_i)$s 
    progress through prior to recovering.}

\add{First, consider $A_0$:
    let $r_0$ be the sequence of transitions that $A_0$ takes in the run $r$.
Let $k_0 \geq 0$ be the smallest natural number such that
    in the $k_0^{\text{th}}$ step of $r_0$, $A_0$ is in the initial state $q_0^0$
    of its recovery $Q_0$.
If $k_0 = 0$, then $A_0$ begins in recovery; 
    but once $A_0$ enters the recovery state $q_0^0$ it will remain in the states
    $S_{Q_0}$ of $Q_0$ forever.
So in this case, $A_0$ simply acts like $Q_0$, and therefore
    $P \parallel \textsc{RDaisy}(Q_0) \parallel A_1 \parallel \dots \parallel A_m$
    can simulate the run $r$.
On the other hand, what if $k_0 > 0$?
Then the first $k_0 - 1$ transitions in $r_0$ are taken within the finite DAG
part of $A_0$; the $k_0^{\text{th}}$ transition goes from the finite DAG portion
to the recovery $q_0^0$; and all the subsequent transitions take place in $Q_0$
beginning at its initial state $q_0^0$.
The first $k_0 - 1$ transitions can all be simulated via self-loops in the daisy
part of the \textsc{RDaisy}$(Q_0)$, e.g.,
$a_j^0 \xrightarrow[]{x?} a_{j + 1}^0$
can be simulated via $d_0^0 \xrightarrow[]{x?} d_0^0$.
The $k_0^{\text{th}}$ transition can be simulated via a transition into recovery,
e.g., $a_{k_0}^0 \xrightarrow[]{y!} q_0^0$
    can be simulated via 
    $d_0^0 \xrightarrow[]{y!} q_0^0$.
Every subsequent transition will be the same, taking place in the recovery $Q_0$.
This suffices to show that
    $P \parallel \textsc{RDaisy}(Q_0) \parallel A_1 \parallel \dots \parallel A_m$
    can simulate the run $r$.
By applying this same argument inductively on the remaining $A_i$s, we find that
\(\textsc{DaisySystem}\)
can simulate~$r$.}

\add{Let~$r'$ be a run of  
\(\textsc{DaisySystem}\)
simulating~$r$.
Then for each $0 \leq i \leq m$, by construction, there exists some~$k_i \geq 0$
such that~$\textsc{RDaisy}(Q_i)$ recovers at step~$k_i$ of~$r_i'$.
Therefore $r' \models \bigwedge_{0 \leq i \leq m} \F \texttt{recover}_i$.
Moreover,~$P$ progresses through the same sequence of states in $r'$ as it does in~$r$.
Since the~\texttt{recover}$_i$ propositions do not occur in the property~$\phi$,
and the daisies with recovery do not have any other atomic propositions,
and the trace induced by $r$ violates~$\phi$,
it therefore follows that the trace induced by $r'$ likewise violates~$\phi$.
Therefore $r'$ violates~$\phi$.
We conclude that $r'$ violates~$\psi$, hence 
\(\textsc{DaisySystem}
\centernot{\models} \psi\)
and we are done.}
\end{proof}

\begin{theorem}
\label{thm:complete}
\add{The solution above is complete, for terminating attackers.}
\end{theorem}
\begin{proof}
\add{First we show that the solution terminates on all inputs.
We assume the model checker terminates, so we always progress to line~2.
If $R = \emptyset$ then we terminate at line~2.
If $R \neq \emptyset$, then we progress to line~3.
Next we loop over $0 \leq i \leq m$.
Each step of the loop terminates in runtime bounded by the length of the shortest
prefix of the violating run ending in recovery for the corresponding attacker
component.  Since the run is assumed to satisfy $\F \texttt{recover}_i$ for each
$0 \leq i \leq m$, this shortest prefix must exist.  Thus the loop completes
in finite time.
Once the loop completes we return the tuple $\vec{A}$ and we are done.
We conclude that the solution terminates on all inputs.}

\add{But, does the solution at least attempt to return an attacker, 
    whenever a terminating attacker exists?
And when no terminating attacker exists, does it state as much?
We will answer both questions affirmatively, beginning with the first.}

\add{Suppose a terminating attacker exists.
In Lemma~\ref{lem:cole}, we proved that 
    terminating attackers induce runs discoverable by daisies with recovery.
It follows that at line~2 of our solution
    the set $R$ will be non-empty.
Thus the solution will at least attempt to return an attacker,
    i.e., it will progress past line~3.}

\add{Now we address the second question.
Suppose no terminating attacker exists.
Then the set $R$ will be empty, and the solution will terminate at line 2.
As both questions are answered in the affirmative, we conclude that the solution
is complete for terminating attackers.}
\end{proof}

%% file: sections/Section5CaseStudyTCP.tex
\techreport{Below we first describe our implementation then the details of our case study (TCP).}{}

\textbf{Implementation}
We implemented our solutions in an open-source tool called 
\textsc{Korg}\techreport{\footnote{Named after the Korg microKORG synthesizer, with its dedicated ``attack" control on Knob 3.  Code and models are freely and openly available at \url{https://github.com/maxvonhippel/AttackerSynthesis}.}}{}.
We say an attacker $\vec{A}$ for a threat model $\textsc{TM} = (P, (Q_i)_{i = 0}^m, \phi)$ is a \emph{centralized attacker} if $m = 0$, or a \emph{distributed attacker}, otherwise.  
In other words, a centralized attacker has only one attacker component $\vec{A} = (A)$, whereas
a distributed attacker has many attacker components $\vec{A} = (A_i)_{i = 0}^m$.
\textsc{Korg} handles $\exists$ASP and R-$\exists$ASP for liveness and safety properties for a centralized attacker.
\textsc{Korg} is implemented in 
\textsc{Python 3} 
and uses the model-checker \textsc{Spin}~\cite{Holzmann03} as its underlying verification engine.

\techreport{
	\begin{table}[H]
	\title{Implemented Solutions}
	\centering
	\begin{tabular}{|r|c|c|c|c|}
	  \hline
	  Attacker: & \multicolumn{2}{c|}{Centralized} & \multicolumn{2}{c|}{Distributed} \\\hline
	  Recovery? & $\exists$ Problem & $\forall$ Problem & $\exists$ Problem & $\forall$ Problem \\\hline
	  With    & \checkmark & \xmark & \xmark & \xmark \\
	  Without & \checkmark & \xmark & \xmark & \xmark \\\hline
	\end{tabular}
	\caption{\checkmark \, or \xmark \, denote if the solution to a problem is or is not implemented in \textsc{Korg}, for both liveness and safety properties.}
	\end{table}
}{}

\textbf{TCP} is a fundamental Internet protocol consisting of three stages: connection establishment, data transfer, and connection tear-down.  We focus on the first and third stages, which jointly we call the connection routine.  Our approach and model (see Fig. \ref{fig:TCPblockDiagram}, \ref{TCPdiagram}) are inspired by \textsc{Snake} \cite{IEEE15}.  Run-times and results are listed in Table~\ref{benchmarksTable}.

\subsubsection{Threat Models}

\begin{figure}\techreport{}{[H]}
\techreport{}{\vspace{-1.5cm}}
\centering
\techreport{
	\begin{adjustbox}{max totalsize={1.0\textwidth}{.7\textheight},center}
}{
	\begin{adjustbox}{max totalsize={0.9\textwidth}{.7\textheight},center}
}
\begin{tikzpicture}
% ------------------------------------------
% |                 BOXES                  |
% ------------------------------------------
% Peer 1
\draw[draw=black] (0, 0) rectangle ++(4, 2.5);
% 1toN
\draw[draw=black] (5, 0) rectangle ++(1.5, 1);
% Nto1
\draw[draw=black] (5, 1.5) rectangle ++(1.5, 1);
% Network
\draw[draw=black,dashed] (7.5, 0) rectangle ++(4, 2.5);
% 2toN
\draw[draw=black] (12.5, 0) rectangle ++(1.5, 1);
% Nto2
\draw[draw=black] (12.5, 1.5) rectangle ++(1.5, 1);
% Peer 2
\draw[draw=black] (15, 0) rectangle ++(4, 2.5);
% ------------------------------------------
% |                 ARROWS                 |
% ------------------------------------------
% Peer 1 -> 1toN
\draw[straight] (4, 0.5) to (5, 0.5);
% 1toN -> Network
\draw[straight] (6.5, 0.5) to (7.5, 0.5);
% Network -> Nto1
\draw[straight] (7.5, 2) to (6.5, 2);
% Nto1 -> Peer 1
\draw[straight] (5, 2) to (4, 2);

% Network -> Nto2
\draw[straight] (11.5, 2) to (12.5, 2);
% Nto2 -> Peer 2
\draw[straight] (14, 2) to (15, 2);
% Peer 2 -> 2toN
\draw[straight] (15, 0.5) to (14, 0.5);
% 2toN -> Network
\draw[straight] (12.5, 0.5) to (11.5, 0.5);
% % ------------------------------------------
% % |                 TEXT                   |
% % ------------------------------------------
% Peer 1
\node[] at (2, 1.25) {\large \textsc{Peer 1}};
% 1toN
\node[] at (5.75, 0.5) {\large \textsc{1toN}};
% Nto1
\node[] at (5.75, 2) {\large \textsc{Nto1}};
% Network
\node[] at (9.5, 1.25) {\large \textsc{Network}};
% 2toN
\node[] at (13.25, 0.5) {\large \textsc{2toN}};
% Nto2
\node[] at (13.25, 2) {\large \textsc{Nto2}};
% Peer 2
\node[] at (17, 1.25) {\large \textsc{Peer 2}};
\end{tikzpicture}
\end{adjustbox}
\caption{TCP threat model block diagram.  Each box is a process.  
An arrow from process $P_1$ to process $P_2$ denotes that a subset of the outputs of $P_2$ are exclusively inputs of $P_1$. 
 \textsc{Peer}s 1 and 2 are TCP peers.
%e.g. laptops on a shared wireless network.  
A \emph{channel} is a directed FIFO queue of size one with the ability to detect fullness.  A full channel may be overwritten.  \textsc{1toN}, \textsc{Nto1}, \textsc{2toN}, and \textsc{Nto2} are channels.  Implicitly, channels relabel: for instance, \textsc{1toN} relabels outputs from \textsc{Peer 1} to become inputs of \textsc{Network}; %for readability, we ignore this subtlety in our diagrams.  
\textsc{Network} transfers messages between peers via channels, and is the vulnerable process.}
\label{fig:TCPblockDiagram}
\techreport{}{\vspace{-1cm}}
\end{figure}

\techreport{We use channels to build asynchronous communication out of direct (rendezvous) communication.}{}
Rather than communicating directly with the \textsc{Network}, the peers communicate with the channels, and the channels communicate with the \textsc{Network}, allowing us to model the fact that packets are not instantaneously transferred in the wild.  We use the shorthand $\textsc{chan} ! \textsc{msg}$ to denote the event where \textsc{msg} is sent over a channel \textsc{chan}; it is contextually clear who sent or received the message.  
\techreport{TCP exists in the \emph{Transport Layer} of the internet, an upper layer reliant on the lower \emph{Link Layer} and \emph{Internet Layer}.}{} 
We abstract the lower network stack layer TCP relies on with \textsc{Network}, which passes messages between $\textsc{1toN} \parallel \textsc{2toN}$ and $\textsc{Nto1} \parallel \textsc{Nto2}$.  We model the peers symmetrically.

\begin{figure}\techreport{}{[H]}
\begin{adjustbox}{max totalsize={.9\textwidth}{.7\textheight},center}
\centering
\begin{tikzpicture}
\def\rectnode[#1]#2(#3)#4[#5]#6 {
	\draw node[
		append after command={[rounded corners=2pt](#3.west)|-(#3.north)},
      	append after command={[rounded corners=2pt](#3.north)-|(#3.east)},
      	append after command={[rounded corners=2pt](#3.east)|-(#3.south)},
      	append after command={[rounded corners=2pt](#3.south)-|(#3.west)},
      	#1
	]#2(#3)#4[#5]#6
}

\node[] (empty) {};

\rectnode[] (closed)  [right=of empty        ] {Closed};
\rectnode[] (end)     [right=of closed       ] {End};
\rectnode[] (listen)  [below=of end          ] {Listen};
\rectnode[] (synsent) [above=of end          ] {\texttt{SYN} Sent};
\rectnode[] (i0)      [above right=of synsent] {$i_0$};
\rectnode[] (i1)      [below=of i0           ] {$i_1$};
\rectnode[] (i2)      [right=of listen       ] {$i_2$};
\rectnode[] (synrec)  [above right=of i2     ] {\texttt{SYN} Received};
\rectnode[] (establ)  [above=of synrec       ] {Established};
\rectnode[] (i3)      [right=of establ       ] {$i_3$};
\rectnode[] (finw1)   [right=of synrec       ] {\texttt{FIN} Wait 1};
\rectnode[] (closew)  [right=of i3           ] {Close Wait};
\rectnode[] (finw2)   [right=of finw1        ] {\texttt{FIN} Wait 2};
\rectnode[] (i4)      [below=of finw1        ] {$i_4$};
\rectnode[] (closing) [below=of listen       ] {Closing};
\rectnode[] (i5)      [right=of i4           ] {$i_5$};
\rectnode[] (timew)   [below=of closing      ] {Time Wait};
\rectnode[] (lastack) [above=of i0           ] {Last \texttt{ACK}};

\draw[straight] (empty) to (closed);
\draw[straight] (closed) to (end);
\draw[looped] (closed) to[out=south east,in=north west] (listen);
\draw[looped,dashed] (listen) to[out=west,in=south] (closed); % a timeout transition
\draw[looped,dashed] (synsent) to[out=south,in=north east] (closed); % a timeout transition
\draw[looped] (closed) to[out=north,in=west,right] node {\textsc{snd}$!$\texttt{SYN}} (synsent);

\draw[straight] (synsent) to[left] node {\textsc{rcv}$?$\texttt{SYN\_ACK}} (i0);
\draw[looped] (synsent) to[out=south east,in=south west,below] node {\textsc{rcv}$?$\texttt{SYN}} (i1);
\draw[looped] (listen) to[out=south east,in=south west,below] node {\textsc{rcv}$?$\texttt{SYN}} (i2);

\draw[straight] (i0) to[right] node {\textsc{snd}$!$\texttt{ACK}} (establ);
\draw[looped] (i1) to[out=south east,in=north west,right] node {\textsc{snd}$!$\texttt{ACK}} (synrec);
\draw[looped] (i2) to[out=north,in=west,below right] node {\textsc{snd}$!$\texttt{SYN\_ACK}} (synrec);
\draw[straight] (synrec) to[right] node {\textsc{rcv}$?$\texttt{ACK}} (establ);

\draw[looped] (establ) to[out=north east,in=north west,above] node {\textsc{rcv}$?$\texttt{FIN}} (i3);
\draw[looped] (establ) to[out=east,in=north west,right] node {\textsc{snd}$!$\texttt{FIN}} (finw1);

\draw[looped] (finw1) to[out=north east,in=north west,above] node {\textsc{rcv}$?$\texttt{ACK}} (finw2);
\draw[straight] (finw1) to[left] node {\textsc{rcv}$?$\texttt{FIN}} (i4);
\draw[looped] (i3) to[out=north east,in=north west,above] node {\textsc{snd}$!$\texttt{ACK}} (closew);

\draw[looped] (closew) to[out=north,in=east,below left] node {\textsc{snd}$!$\texttt{FIN}} (lastack);
\draw[looped] (i4) to[out=south,in=east,above] node {\textsc{snd}$!$\texttt{ACK}} (closing);
\draw[straight] (finw2) to[right] node {\textsc{rcv}$?$\texttt{FIN}} (i5);
\draw[looped] (i5) to[out=south,in=east,above] node {\textsc{snd}$!$\texttt{ACK}} (timew);

\draw[looped] (lastack) to[out=west,in=north,right] node {\textsc{rcv}$?$\texttt{ACK}} (closed);
\draw[looped] (closing) to[out=west,in=south,right] node {\textsc{rcv}$?$\texttt{ACK}} (closed);
\draw[looped] (timew)   to[out=west,in=south] (closed);
\end{tikzpicture}
\end{adjustbox}
\caption{A TCP peer.  For $i = 1, 2$, if this is $\textsc{Peer }i$, then $\textsc{snd} := i\textsc{toN}$ and $\textsc{rcv} := \textsc{Nto}i$.  
All the states except $i_0, ..., i_5,$ and End are from the finite state machine in the TCP RFC \cite{TCPRFC}.  The RFC diagram omits the implicit states $i_0, ..., i_5$, instead combining send and receive events on individual transitions.  In the RFC, Closed is called a ``fictional state", where no TCP exists.  We add a state End to capture the difference between a machine that elects not to instantiate a peer and a machine that is turned off.  We label each state $s$ with a single atomic proposition $s_i$.  Dashed transitions are \emph{timeout transitions}, meaning they are taken when the rest of the system deadlocks. \label{TCPdiagram}}
\techreport{}{\vspace{-1cm}}
\end{figure}

Given a property $\phi$ about TCP, we can formulate a threat model $\textsc{TM}$ as follows, where we assume the adversary can exploit the lower layers of a network and ask if the adversary can induce TCP to violate $\phi$: 
\begin{equation}
\textsc{TM} = (\textsc{Peer 1} \parallel
			   \textsc{Peer 2} \parallel
			   \textsc{1toN} \parallel
			   \textsc{2toN} \parallel
			   \textsc{Nto1} \parallel
			   \textsc{Nto2},
			   (\textsc{Network}),
			   \phi)
\end{equation}
We consider the properties $\phi_1, \phi_2, \phi_3$, giving rise to the threat models $\textsc{TM}_1, \textsc{TM}_2, \textsc{TM}_3$.

\subsubsection{$\textsc{TM}_1$: No Half-Closed Connection Establishment}
The safety property $\phi_1$ says that if \textsc{Peer 1} is in Closed state, then \textsc{Peer 2} cannot be in Established state.
\begin{equation}
\phi_1 = \G ( \text{Closed}_1 \implies \neg \text{Established}_2 )
\end{equation}
\textsc{Korg} discovers an attacker that spoofs the active participant in an active-passive connection establishment (see message sequence chart in Fig. \ref{fig:EattackerTCPMSC}), as described in \cite{SNI97}.
%with a \emph{message sequence chart} below.  
\techreport{Note that our model does not capture message sequence numbers and 
in the real world the attacker also needs to guess the sequence number of the passive peer.}{}

\begin{figure}[H]
\begin{adjustbox}{max totalsize={.7\textwidth}{.4\textheight},center}
\centering
\begin{tikzpicture}
%%%%%%%%%%%%%%%%%%%%%%%%%%%%%%%%%%%%%%%%%%%%%%%%%%%%%%%%%%%%%%%
% 	PEER 1                 ATTACKER                 PEER 2    %            
%%%%%%%%%%%%%%%%%%%%%%%%%%%%%%%%%%%%%%%%%%%%%%%%%%%%%%%%%%%%%%%
%   CLOSED                                          CLOSED
%   CLOSED      ack------                           LISTEN
%   CLOSED        |                 -------syn----> LISTEN
%   CLOSED        |                 <----syn-ack--- SYN_RECEIVED
%   CLOSED        |                 -------ack----> ESTABLISHED
% property is violated
% 
\draw[draw=gray,dashed] (2.2, 6.5) rectangle ++(1.6, 3.5);
\draw[draw=gray,dashed] (6.2, 6.5) rectangle ++(1.6, 3.5);
\draw[draw=gray,dashed] (11,   6.5) rectangle ++(2,   3.5);

\node[draw,rectangle] (peer1) at (3 ,9.5) {\textsc{Peer 1}};
\node[draw,rectangle] (attac) at (7 ,9.5) {$A$};
\node[draw,rectangle] (peer2) at (12,9.5) {\textsc{Peer 2}};

\node[] (closed10) at (3 ,9) {Closed};
\node[] (a00)      at (7 ,9) {$a_0$};
\node[] (closed20) at (12,9) {Closed};

\node[] (closed11) at (3 ,8.5) {Closed};
\node[] (a01)      at (7 ,8.5) {$a_1$};
\node[] (listen20) at (12,8.5) {Listen};

\node[] (closed12) at (3 ,8) {Closed};
\node[] (a02)      at (7 ,8) {$a_2$};
\node[] (listen21) at (12,8) {Listen};

\draw[straight] (a01) to node[above] {\texttt{ACK}} (closed11);
\draw[straight] (a02) to node[above] {\texttt{SYN}} (listen21);

\node[] (closed13) at (3 ,7.5) {Closed};
\node[] (a03)      at (7 ,7.5) {$a_3$};
\node[] (synRec20) at (12,7.5) {\texttt{SYN} Received};

\draw[straight] (synRec20) to node[above] {\texttt{SYN\_ACK}} (a03);

\node[] (closed13) at (3 ,7) {Closed};
\node[] (a03)      at (7 ,7) {$a_3$};
\node[] (established20) at (12,7) {Established};

\draw[straight] (a03) to node[above] {\texttt{ACK}} (established20);
	\end{tikzpicture}
\end{adjustbox}
\caption{Time progresses from top to bottom.
Labeled arrows denote message exchanges over implicit channels.  The property is violated in the final row; after this recovery may begin.  \add{Note that the initial \texttt{ACK} injected to Peer 1 is unnecessary.  Because \textsc{Spin} attempts to find counterexamples as quickly as possible, the counterexamples it produces are not in general minimal.}}
\label{fig:EattackerTCPMSC}
\end{figure}

\subsubsection{$\textsc{TM}_2$: Passive-Active Connection Establishment Eventually Succeeds}
The liveness property $\phi_2$ says that if it is infinitely often true that \textsc{Peer 1} is in Listen state while \textsc{Peer 2} is in \texttt{SYN} Sent state, then it must eventually be true that \textsc{Peer 1} is in Established state.
\begin{equation}
\phi_2 = (\G \F ( \text{Listen}_1 \land \text{\texttt{SYN} Sent}_2 ) )
		\implies \F\, \text{Established}_1
\end{equation}
\textsc{Korg} discovers an attack where 
\add{an \texttt{ACK} packet is injected to each peer, and}
a \texttt{SYN} packet from \textsc{Peer 2} is dropped.
The corresponding attacker code is given in the \textsc{Promela} language of  \textsc{Spin}
 in Fig.  \ref{fig:EattackerFiniteTM2code}.  \techreport{The attacker in Fig. \ref{fig:EattackerFiniteTM2code} also induces deleterious behavior not captured by violation of $\phi_2$, where the system deadlocks in \rem{$(\texttt{SYN}\text{ Sent}, \texttt{SYN}\text{ Received})$}\add{$(\text{Listen}, \texttt{SYN}\text{ Sent})$.}}

\techreport{
	\begin{figure}[H]
}{
	\begin{figure}
}
\centering
\begin{Verbatim}[tabsize=4]
Nto1 ! ACK; Nto2 ! ACK; 2toN ? SYN;
\end{Verbatim}
\caption{Body of \textsc{Promela} process for a TM$_2$-attacker \rem{with recovery }generated by \textsc{Korg}.  \add{The attacker injects an \texttt{ACK} packet to each peer.}  \textsc{Peer~2} transitions from Closed state to \texttt{SYN} Sent state and sends \texttt{SYN} to \textsc{Peer 1}.  The attacker drops this packet so that it never reaches \textsc{Peer~1}.  \add{Neither peer can process an \texttt{ACK} from its current state, and the system deadlocks.}  \rem{\textsc{Peer~1} then transitions back and forth forever between Closed and Listen states, and the property is violated.  Because \textsc{Spin} attempts to find counterexamples as quickly as possible, the counterexamples it produces are not in general minimal.}}
\label{fig:EattackerFiniteTM2code}
% \techreport{}{\vspace{-1cm}}
\end{figure}

%%%%%%%%%%%%%%%%%%%%%%%%%%%%%%%%%%%%%%%%%%%%%%%%%%%%%%%%%%%%%%%
% 	PEER 1                 ATTACKER                 PEER 2    %            
%%%%%%%%%%%%%%%%%%%%%%%%%%%%%%%%%%%%%%%%%%%%%%%%%%%%%%%%%%%%%%%
%  CLOSED  /--syn----                                CLOSED
%  CLOSED  |                                         LISTEN
%  CLOSED  |                          -----syn-----> LISTEN
%  CLOSED  |                          <---syn-ack--- SYN_RECEIVED
%  LISTEN </
%  SYN_RECEIVED -- syn-ack->
% deadlock in SYN_RECEIVED x SYN_RECEIVED
\subsubsection{$\textsc{TM}_3$: Peers Do Not Get Stuck}
The safety property $\phi_3$ says that the two peers will never simultaneously deadlock outside their End states.  Let $S_i$ denote the set of states in Fig. \ref{TCPdiagram} for \textsc{Peer} $i$, and $S_i' = S_i \setminus \{ \text{End} \}$.
\begin{equation}
\phi_3 = \bigwedge_{s_1 \in S_1'} \bigwedge_{s_2 \in S_2'} \neg \F \G (s_1 \land s_2)
\end{equation}
For the problem with recovery, \textsc{Korg} discovers an attacker that \rem{selectively drops the \texttt{ACK} sent by \textsc{Peer 1} as it transitions from $i_0$ to Established state in an active/passive connection establishment routine, leaving \textsc{Peer 2} stranded in \texttt{SYN} Received state, leading to a violation of $\phi_3$}\add{injects a \texttt{SYN} packet to each peer, causing both to transition into \texttt{SYN} Received, where they get stuck (violating $\phi_3$)}.  
\rem{
\techreport{
	This type of packet-loss-induced deadlock is not uncommon in 
	real-world TCP implementations, e.g. 
	\cite{EE06,SE14,Arch07}.
}{
	Similar bugs exist in real-world implementations, e.g. \cite{Arch07}.
}}
\techreport{}{\vspace{-1em}}

\subsubsection{Performance}

Performance results for Case Study are given in Table~\ref{benchmarksTable}.  
\techreport{
	The discovered attackers either implement known attacks or reproduce known bugs in real TCP implementations.  
}{}
Our success criteria was to produce realistic attackers faster than an 
expert human could with pen-and-paper.  We discovered attackers in seconds
or minutes as shown in Table~\ref{benchmarksTable}.

\techreport{
	\begin{table}[H]
}{
	\begin{table}
}
	\title{Case Study Run-times \& Results}
	\centering
	\begin{tabular}{|c|c|c|c|c|}
	\hline
	\multirow{2}{*}{Property} & 
	\multicolumn{2}{|c|}{$\cfrac{\text{Avg. Runtime (s)}}{\text{Unique Attacker}}$} & 
	\multicolumn{2}{|c|}{Unique Attackers Found} \\
	\cline{2-5}
	          & $\exists$ASP & R-$\exists$ASP & $\exists$ASP & R-$\exists$ASP \\
	\hline
	$\phi_1$ & \rem{0.32}\add{0.43}     & \rem{0.49}\add{3.39}       & 7 & 1 \\
	$\phi_2$ & \rem{0.45}\add{0.75}     & \rem{0.48}\add{3.29}       & 4 & 1 \\
	$\phi_3$ & \rem{876.74}\add{122.88} & \rem{2757.98}\add{2469.74} & 3 & 1 \\\hline 
	\end{tabular}
\techreport{
\caption{
\add{Since we use size-1 channels, we allow \textsc{Korg} to configure \textsc{Spin}
with partial order reduction turned on, which improves overall runtime.}
For each property $\phi_i$, we asked \textsc{Korg} to generate
ten attackers with recovery, and ten without\rem{.  
We repeated this experiment ten times}, \add{using} 
\rem{on} a 16Gb 20\rem{18}\add{20} quad-core Intel\copyright \, Core$^{\text{tm}}$ 
i7-\rem{8550U}\add{1185G7} CPU running Linux Mint \rem{19.3}\add{20.1} Cinnamon.
\textsc{Korg} may generate duplicate attackers, so for   
each property (Column~1), 
we list the \rem{average }time taken to generate a unique attacker 
without recovery (Column~2) or with recovery (Column~3), and
the total number of unique attackers found without recovery (Column~4)
or with recovery (Column~5).
For example, for $\phi_3$, out of \rem{100}\add{10} 
attackers with recovery generated over the
course of about \rem{12 hours}\add{41 minutes}, \rem{five were}\add{one was} 
unique and \rem{95}\add{the rest were} duplicates\rem{,
meaning \textsc{Korg} took on average about 15 minutes per unique attacker}.
Instructions and code to reproduce these results are given in the GitHub 
repository.}
}{
\caption{
For each property $\phi_i$, we asked \textsc{Korg} 10 times to generate
10 attackers with recovery, and 10 without, 
on a 16Gb 2018 Intel\copyright \, Core$^{\text{tm}}$ 
i7-8550U CPU running Linux Mint 19.3 Cinnamon.
\textsc{Korg} may generate duplicate attackers, so for   
each property (Column 1), 
we list the average time to generate a unique attacker 
without recovery (Column 2) or with (Column 3), and
the total number of unique attackers found without recovery (Column 4)
or with (Column 5).
E.g., for $\phi_3$, of 100 attackers with recovery generated over about four hours, 
five were unique and 95 duplicates,
so \textsc{Korg} took about 2.3 minutes per attacker, or, 
45 minutes per unique attacker. 
Instructions and code to reproduce these results are given in the GitHub 
repository.}	
}
\label{benchmarksTable}
\techreport{}{\vspace{-1.2em}}
\end{table}

\techreport{

We chose TCP connection establishment for our case study because it is simple and well-understood.  Across three properties (two safety and one liveness), with and without recovery, \textsc{Korg} synthesized attackers exhibiting attack strategies that have worked or could work against some real-world TCP implementations, modulo sequence numbers.  The synthesized attackers are simple, consisting of only a few lines of code, but our TCP model is also simple since we omitted sequence numbers, congestion control, and other details.  Moving forward, we want to apply \textsc{Korg} to more complicated models and discover novel exploits.}{}

%% file: sections/Section6RelatedWork.tex
Prior works formalized security problems using game theory (e.g., \textsc{FlipIt} \cite{FlipIt}, \cite{klavska2018automatic}), ``weird machines" \cite{USENIX11}, attack trees \cite{ACMCSUR19}, Markov models \cite{valizadeh2019toward}, and other methods.  
%Using automata is not novel, although our particular formalization is.  
Prior notions of attacker quality include $\mathcal{O}$-complexity \cite{LangSec18}, expected information loss \cite{Citeseer13}, or success probability \cite{meira2019synthesis,vasilevskaya2014}, which is similar to our concept of $\forall$ versus $\exists$-attackers.  The formalism of \cite{vasilevskaya2014} also captures attack consequence (cost to a stakeholder).

\techreport{
	Nondeterminism abstracts probability, e.g., a $\forall$-attacker 
	is an attacker with $P(\text{success}) = 1$, and, under fairness conditions, 
	an $\exists$-attacker is an attacker with $0 < P(\text{success}) < 1$.  
	Probabilistic approaches are advantageous when the existence of an
	event is less interesting than its likelihood.
	For example, a lucky adversary \emph{could} randomly guess my RSA modulus, 
	but this attack is too unlikely to be interesting.
	We chose to use nondeterminism over probabilities for two reasons: first, 
	because nondeterministic models do not require prior knowledge of event
	probabilities, but probabilistic models do; and second, because the 
	non-deterministic model-checking problem is cheaper than its probabilistic
	cousin \cite{vardi1999probabilistic}.
	Nevertheless, we believe our approach could be extended to probabilistic
	models in future work.
	Katoen provides a nice survey of probabilistic model checking \cite{katoen2016probabilistic}.
}{}

Attacker synthesis work exists in cyber-physical systems \techreport{\cite{phan2017synthesis,bang2018online,huang2018algorithmic,lin2019synthesis,meira2019synthesis}}{\cite{phan2017synthesis,bang2018online,huang2018algorithmic,meira2019synthesis}}, 
most of which define attacker success using bad states 
(e.g., reactor meltdown, vehicle collision, etc.) 
or information theory 
(e.g., information leakage metrics).
Problems include 
the \emph{actuator attacker synthesis problem} \cite{lin2019synthesisActuator};
the \emph{hardware-aware attacker synthesis problem} \cite{trippel2019security};
and the \emph{fault-attacker synthesis problem} \cite{barthe2014synthesis}.
\techreport{Defensive synthesis also exists~\cite{ardeshiricham2019verisketch}.}{} 

Maybe the most similar work to our own is
\textsc{ProVerif} \cite{BlanchetCSFW01}, which verifies properties of, and 
generates attacks against, cryptographic protocols.  We formalize the problem 
with operational semantics (processes) and reduce it to model checking, whereas 
\textsc{ProVerif} uses axiomatic semantics (\textsc{ProLog} clauses) and reduces 
it to automated proving.
Another similar tool is \textsc{NetSMC} \cite{NetSMC}, 
a model-checker that efficiently finds counter-examples to security properties 
of stateful networks.
\techreport{\textsc{NetSMC} is a model-checker, whereas our approach is built on 
top of a model checker, so in theory our approach could be implemented on top of 
\textsc{NetSMC}.}{}

Existing techniques for automated attack discovery include 
model-guided search \cite{IEEE15,hoque2017analyzing} 
(including using inference \cite{cho2011mace}), 
open-source-intelligence \cite{SemFuzz}, 
bug analysis \cite{huang2012crax}, 
and genetic programming \cite{kayacik2009generating}.  
The generation of a failing test-case for a protocol property 
is not unlike attack discovery, 
so \cite{mcmillan2019formal} is also related.

\techreport{
TCP was previously formally studied using a process language called \textsc{SPEX} \cite{schwabe1981formal}, Petri nets \cite{han2005termination}, the \textsc{HOL} proof assistant \cite{bishop2018engineering}, and various other algebras (see Table~2.2 in \cite{Thesis2016}).  Our model is neither the most detailed nor the most comprehensive, but it captures all possible establishment and tear-down routines, and is tailored to our framework.}{}

This paper focuses on attacker synthesis at the protocol level, and thus
differs from the work reported in~\cite{KangLafortuneTripakisCAV2019}
in two ways:
first, the work in~\cite{KangLafortuneTripakisCAV2019} 
synthesizes mappings between high-level protocol models and execution
platform models, thereby focusing on linking protocol design and
implementation;
second, the work in~\cite{KangLafortuneTripakisCAV2019} synthesizes
correct (secure) mappings, whereas we are interested in synthesizing
attackers.

%% file: sections/Section7Conclusion.tex
We present a novel formal framework for automated attacker synthesis.
The framework includes an explicit definition of threat models and four novel,
to our knowledge,\techreport{ abstract }{ }categories of attackers. We formulate   
four attacker synthesis problems, and propose solutions to two of them by 
program transformations and reduction to model-checking.
We prove our solutions \add{are} sound\add{.}\rem{ and complete; }
\add{We prove our solution to the first problem is complete generally, and that our solution to the second problem is complete over a restricted class of attackers.}
\techreport{\rem{sketches of these proofs}\add{Proofs} are provided in Section~\ref{sec:sol}.}{\rem{these proofs}\add{Proofs} are available online \cite{von2020automated}.}
Finally, we implement our solutions for the case of a centralized attacker in an open-source tool called \textsc{Korg}, apply \textsc{Korg} to the study of the TCP connection routine, and discuss the results. 
\textsc{Korg} and the TCP case study 
are freely and openly 
available\techreport{, so our results are easy to reproduce}
{\footnote{{\url{github.com/maxvonhippel/AttackerSynthesis}}}}.

%% file: sections/Section8Erratum.tex
\add{In a previous version of this paper~\cite{hippel2020automated}, 
we proposed an incorrect solution to the R-$\exists$ASP,
which employed a model-checker to find a violating 
run~$r$ described as a lasso~$\alpha \cdot \beta^{\omega}$,
then let each attacker component~$A_i$ consist of
a process reproducing the component of~$\alpha$ in the interface
of~$Q_i$, followed by the recovery~$Q_i$.  
That solution was flawed because it
assumed that at the end of~$\alpha$ (no later and no sooner), all of the
daisies with recovery would recover.
We illustrate how that assumption could fail below.}

\add{Consider the processes $Q$ and $P$ in Fig.~\ref{fig:PQerr}, and let~$\textsc{TM}=(P,(Q),\phi)$ be a threat model where $\phi=\G \neg \texttt{dead}$.
A lasso run $r = \alpha \cdot \beta^{\omega}$ of $P \parallel \textsc{RDaisy}(Q)$ violating $\psi = (\F \texttt{recover}) \implies (\G \neg \texttt{dead})$ is given in Fig.~\ref{eqn:PQerrRun}.  Since at least one such run exists, in both our prior (incorrect) and revised (correct) solutions, $R \neq \emptyset$.}

\begin{figure}[H]
\centering
\begin{tikzpicture}
\node[] (empty) at (-2, 0) {};
\node[draw,circle] (q0) at (0, 0) {$q_0$};
\node[draw,circle] (q1) at (2, 1) {$q_1$};
\node[draw,circle] (q2) at (2,-1) {$q_2$};
\draw[straight] (empty) to (q0);
\draw[straight] (q0) to node[above left,rotate=30,xshift=0.6cm] 
{$\texttt{BAD}?$} (q1);
\draw[straight] (q0) to node[below left,rotate=-30,xshift=0.7cm] 
{$\texttt{GOOD}?$} (q2);
\draw[looped] (q1) to[out=north east,in=south east,looseness=5] node[right] {\texttt{DOBAD}$!$} (q1);
\draw[looped] (q2) to[out=north east,in=south east,looseness=5] node[right] {\texttt{TRIGGER}$!$} (q2);
\node[] (emptyP) at (7, 0) {};
\node[draw,circle] (p0) at (9, 0) {$p_0$};
\node[draw,circle] (p1) at (9,-2) {$p_1$};
\node[draw,circle] (p2) at (13,-1) {$p_2$};
\node[draw,rounded rectangle] (p3) at (13, 1) {$p_4 : \{ \texttt{dead} \}$};
\draw[straight] (emptyP) to (p0);
\draw[straight] (p0) to node[above,rotate=15] {\texttt{DOBAD}$?$} (p3);
\draw[straight] (p0) to node[below,rotate=-15] {\texttt{GOOD}$!$} (p2);
\draw[looped] (p0) to[out=south west,in=north west,looseness=1]
    node[left] {\texttt{TRIGGER}$?$} (p1);
\draw[looped] (p1) to[out=north east,in=south east,looseness=1]
    node[left] {\texttt{BAD}$!$} (p0);
\draw[looped] (p3) to[out=north west,in=north east,looseness=2]
    node[above] {\texttt{DOBAD}$?$} (p3);
\draw[looped] (p2) to[out=south west,in=south east,looseness=5]
    node[below] {\texttt{TRIGGER}$?$} (p2);
\end{tikzpicture}
\caption{Left: the process $Q$.  Notice that $Q$ is deterministic.  Right: the process $P$.  Observe that $P \parallel Q \models \G \neg \texttt{dead}$.}
\label{fig:PQerr}
\end{figure}

\begin{figure}[H]
\[
\underbrace{(p_0, d_0) \xrightarrow[]{\texttt{TRIGGER}!}
(p_1, q_0) \xrightarrow[]{\texttt{BAD}!}
(p_0, q_1) \xrightarrow[]{\texttt{DOBAD}!}}_{\alpha}
\underbrace{(p_4, q_1) \xrightarrow[]{\texttt{DOBAD}!} (p_4, q_1)
    \xrightarrow[]{\texttt{DOBAD}!} \dots}_{\beta^{\omega}}
\]
\caption{An example (lasso) run $r$ of $P \parallel \textsc{RDaisy}(Q)$ violating $\psi = (\F \texttt{recover}) \implies (\G \neg \texttt{dead})$.}
\label{eqn:PQerrRun}
\end{figure}

\add{Our prior (incorrect) solution to the \textsc{R$\exists$ASP} would produce the following process,~$A$.}

\begin{figure}[H]
\centering
\begin{tikzpicture}
\node[] (empty) at (-10, 0) {};
\node[draw,circle] (a0) at (-8, 0) {$a_0$};
\node[draw,circle] (a1) at (-6, 0) {$a_1$};
\node[draw,circle] (a2) at (-4, 0) {$a_2$};
\draw[straight] (empty) to (a0);
\node[draw,rounded rectangle] (q0) at (0, 0) 
    {$q_0 : \{ \texttt{recover} \}$};
\draw[looped] (a0) to[out=north east,in=north west,looseness=1] 
    node[above] {\texttt{TRIGGER}$!$} (a1);
\draw[looped] (a1) to[out=north east,in=north west,looseness=1] 
    node[above] {\texttt{BAD}$?$} (a2);
\draw[straight] (a2) to
    node[above] {\texttt{DOBAD}$!$} (q0);
\node[draw,circle] (q1) at (2, 1) {$q_1$};
\node[draw,circle] (q2) at (2,-1) {$q_2$};
\draw[straight] (q0) to node[above left,rotate=30,xshift=0.6cm] 
{$\texttt{BAD}?$} (q1);
\draw[straight] (q0) to node[below left,rotate=-30,xshift=0.7cm] 
{$\texttt{GOOD}?$} (q2);
\draw[looped] (q1) to[out=north east,in=south east,looseness=5] node[right] {\texttt{DOBAD}$!$} (q1);
\draw[looped] (q2) to[out=north east,in=south east,looseness=5] node[right] {\texttt{TRIGGER}$!$} (q2);
\end{tikzpicture}
\caption{The ``attacker'' component produced by our prior (incorrect) solution given the example threat model.}
\end{figure}

\add{But critically,~$P \parallel A$ has no runs,
hence~$A$ is not actually a \textsc{TM}-attacker.
The problem is that we created the attacker program by 
    projecting the prefix~$\alpha$ onto the interface of~$Q$, 
    and gluing the result directly to the recovery~$Q$.
Instead, we should have projected the prefix of~$r$ up until \texttt{recover},
    and glued that result to recovery.
This is exactly what we do in our updated solution, in this case, yielding the  terminating attacker $A'$ illustrated in Fig.~\ref{fig:fixA}.  The attacked system~$P \parallel A'$ yields the run \[r' = (p_0,a_0') \xrightarrow[]{\texttt{TRIGGER}!}
(p_1,q_0) \xrightarrow[]{\texttt{BAD}!}
(p_0,q_1) \xrightarrow[]{\texttt{DOBAD}!}
(p_4,q_1) \xrightarrow[]{\texttt{DOBAD}!}
(p_4,q_1) \xrightarrow[]{\texttt{DOBAD}!} \dots\] which violates $\psi = (\F \texttt{recover}) \implies (\G \neg \texttt{dead})$.}

\begin{figure}[H]
\centering
\begin{tikzpicture}
\node[] (empty) at (-6, 0) {};
\node[draw,circle] (a0) at (-4, 0) {$a_0'$};
\draw[straight] (empty) to (a0);
\draw[straight] (a0) to node[above] {\texttt{TRIGGER}$!$} (q0);
\node[draw,rounded rectangle] (q0) at (0, 0) {$q_0 : \{ \texttt{recover} \}$};
\node[draw,circle] (q1) at (2, 1) {$q_1$};
\node[draw,circle] (q2) at (2,-1) {$q_2$};
\draw[straight] (q0) to node[above left,rotate=30,xshift=0.6cm] 
{$\texttt{BAD}?$} (q1);
\draw[straight] (q0) to node[below left,rotate=-30,xshift=0.7cm] 
{$\texttt{GOOD}?$} (q2);
\draw[looped] (q1) to[out=north east,in=south east,looseness=5] node[right] {\texttt{DOBAD}$!$} (q1);
\draw[looped] (q2) to[out=north east,in=south east,looseness=5] node[right] {\texttt{TRIGGER}$!$} (q2);
\end{tikzpicture}
\caption{The corrected attacker component~$A'$ produced by our updated solution given the example threat model~\textsc{TM}.}
\label{fig:fixA}
\end{figure}